\documentclass[12pt]{article}
\usepackage{float}

\usepackage[T1]{fontenc}
\usepackage{amsmath,amsthm,amssymb,amsfonts}
\usepackage{graphicx}
\usepackage{hyperref}
\usepackage{natbib}
\usepackage{geometry}
\usepackage{setspace}
\usepackage{tikz}
\usepackage{pgfplots}
\pgfplotsset{compat=1.17}
\usepackage{algorithm}
\usepackage{algorithmic}

\newtheorem{theorem}{Theorem}
\newtheorem{proposition}{Proposition}

\newtheorem{assumption}{Assumption}

\newcommand{\E}{\mathbb{E}}
\newcommand{\Var}{\text{Var}}

\title{\textbf{Exploration Is Not What It Seeks:} \\ 
\Large Catalytic Exploration under Status Quo Uncertainty\thanks{This is the extended working paper version containing additional proofs, structural estimation details, and theoretical extensions.}}

\author{
    Zeyu He \\
    Tsinghua University \\
    University of California, Berkeley
}

\date{\today}

\begin{document}

\maketitle

\begin{abstract}
We identify a distinct motive for search, termed catalytic exploration, where agents rationally explore alternatives they expect to reject to resolve uncertainty about the status quo. By decomposing option value into switching and catalytic components, we show that high exploration rates can coexist with bounded switching probabilities. This mechanism generates three insights. First, strong catalytic motives cause separating equilibria to collapse in signaling games as receivers explore indiscriminately. Second, agents optimally acquire more precise information about the status quo than about alternatives, reversing rational inattention intuitions. Third, catalytic exploration creates negative externalities: information technology improvements can paradoxically reduce welfare by encouraging excessive benchmarking.

\vspace{0.3cm}
\noindent \textbf{Keywords:} Catalytic exploration, Search theory, Status quo uncertainty, Endogenous information, Signaling, Reflexive choice

\noindent \textbf{JEL Classification:} D81, D82, D83, D84, L15
\end{abstract}
\newpage

\section{Introduction}

Economic agents routinely engage in costly exploration of alternatives they ultimately reject. Consider, for instance, an individual in a committed relationship who harbors lingering doubts not only about the partner's true character but also about his or her own subjective compatibility with the relationship. This individual may choose to date others, even those appearing ex-ante inferior, not with the intent to switch, but to establish a comparative benchmark. Through these interactions, the agent resolves the uncertainty surrounding the status quo, often leading to a firmer choice of the original partner rather than a breakup. Similarly, job seekers interview for positions below their reservation wage, consumers research products they never purchase, and firms investigate technologies they will not adopt. 

Standard search models, from \cite{stigler1961economics} to contemporary formulations, cannot rationalize this behavior: if an alternative is known to be inferior ex ante, exploration appears to represent pure deadweight loss from the perspective of standard adoption motives.

We propose that exploration serves a dual informational function. Beyond learning about alternatives—the traditional channel—exploration provides a comparative benchmark that resolves uncertainty about one's current situation. This indirect channel, which we term catalytic value, can rationalize exploration of ex-ante inferior options.

Consider an agent choosing between a status quo with payoff $\mu_0 + \epsilon$ and a challenger with quality $\theta \sim F(\cdot)$ having mean $\mu_1 < \mu_0$. The parameter $\epsilon$ represents the unknown match value between the agent's preferences and the status quo option—a subjective valuation component that becomes observable only through comparison. We do not assume restrictive information constraints but rather build on the consensus that valuation is inherently reference-dependent. In this context, the outside option naturally serves as the reference point required to crystallize the value of the status quo. Without exploration, the agent chooses based on expected values, obtaining $\max\{\mu_0 + \E[\epsilon], \mu_1\} = \mu_0$ (since $\E[\epsilon]=0$ and $\mu_1 < \mu_0$). With exploration, she observes both $\epsilon$ and $\theta$, obtaining $\E[\max\{\mu_0 + \epsilon, \theta\}]$.

The total option value from exploration naturally decomposes into two components. The switching value captures the traditional gain from potentially finding a superior alternative: the difference between the expected maximum when $\theta$ is random versus when it equals its mean $\mu_1$. The catalytic value captures the gain from resolving status quo uncertainty through comparison: the difference between the expected maximum when $\epsilon$ is observable versus when only its mean is known. Crucially, this catalytic value is strictly positive whenever $\sigma_\epsilon > 0$, even when $\mu_1 < \mu_0$ ensures negligible adoption probability.

Under normality, we derive an exact formula for catalytic value as a function of status quo uncertainty $\sigma_\epsilon$ and the quality gap $\Delta = \mu_0 - \mu_1$. In the high-uncertainty limit, catalytic value grows linearly with $\sigma_\epsilon$. This growth reflects the fundamental convexity of option value: since the agent can switch if the status quo proves detrimental, higher uncertainty expands the upside potential without exposing the agent to commensurate downside risk. 

The catalytic mechanism fundamentally alters equilibrium predictions in three domains. First, in signaling games, strong catalytic motives destroy the single-crossing property that enables separation. When receivers explore primarily for catalytic reasons, high-quality senders cannot distinguish themselves through costly signals because low types are explored regardless. This provides a new channel for market unraveling distinct from adverse selection.

Second, with endogenous information acquisition, catalytic motives reverse the standard logic of attention allocation. When exploration primarily serves to evaluate the status quo, agents optimally acquire more precise signals about their current situation than about alternatives. This contrasts sharply with rational inattention models where cognitive resources flow toward the most uncertain or highest-variance options.

Third, catalytic exploration generates a previously unrecognized externality. When individuals use others as informational benchmarks without compensation, they impose costs that standard Pigouvian taxes on search activity can address. The optimal policy involves taxing exploration while subsidizing successful matches, with the tax rate increasing in the catalytic-to-switching value ratio.

\subsection{Related Literature}

Our theory contributes to three distinct strands of literature by identifying a reflexive motive for information acquisition that operates through comparative evaluation rather than direct learning.

\subsubsection{Search and Experimentation}

The canonical search paradigm relies on an instrumental logic: agents explore to discover options that stochastically dominate their status quo. The seminal contribution of \cite{weitzman1979optimal} characterizes optimal sequential search through a reservation value property, where exploration continues until the best discovered option exceeds the value of further search. This framework has been extended to ordered search \citep{armstrong2017ordered}, simultaneous search \citep{chade2006simultaneous}, and consumer search with price competition \citep{choi2018consumer}. Recent theoretical advances include \cite{pastrian2025full} who examines how consideration sets affect surplus extraction in search environments—yet throughout these models, exploration value derives exclusively from the possibility of adoption.

We depart from this paradigm by identifying a source of option value—catalytic value—that persists even when adoption probability is negligible. The mechanism differs fundamentally from learning-by-experimenting models \citep{bolton1999strategic, keller2005strategic} where agents must consume to learn their type. In our framework, learning occurs through comparative inspection without consumption, explaining why extensive exploration often precedes decision persistence rather than switching. This insight extends recent empirical findings by \cite{agarwal2024searching} who document substantial search activity in mortgage markets with minimal switching, suggesting exploration serves purposes beyond finding better options.

The distinction matters for welfare analysis. In standard search models, reducing search costs unambiguously improves welfare by enabling more efficient matching. Under catalytic exploration, lower search costs can reduce welfare by encouraging excessive benchmarking that imposes negative externalities on those being explored. This reversal occurs precisely when the ratio of catalytic to switching value exceeds a threshold we characterize.

\subsubsection{Information Acquisition and Attention Allocation}

Our results on endogenous information structure converse with the rational inattention literature pioneered by \cite{sims2003implications}. Standard models of costly information acquisition \citep{matejka2015rational, caplin2015revealed} predict that agents focus cognitive resources on the most uncertain or highest-variance alternatives. The recent contribution by \cite{miao2024dynamic} extends this framework to dynamic discrete choice under rational inattention, maintaining the principle that attention flows toward uncertainty.

We establish an "asymmetric attention" result that reverses this intuition. When catalytic motives dominate (formally, when $V^{ISQ} > V^{IC}$), agents optimally allocate higher precision to signals about the status quo than about alternatives. The mechanism operates through the complementarity between status quo information and the option to explore: better knowledge of one's current situation increases the value of using alternatives as benchmarks.

This finding bridges the economics literature with psychological theories of motivated information acquisition. While \cite{benabou2002self} and \cite{golman2017information} emphasize ego-utility and self-confidence motives for information avoidance, we provide a purely instrumental foundation. Our mechanism also differs from the reference effects studied by \cite{cerreia2024caution}, who axiomatize cautious behavior through reference-dependent preferences. In contrast, our reference effects emerge endogenously from the informational value of comparison, without assuming any behavioral preferences or motivated reasoning.

\subsubsection{Signaling and Market Unraveling}

In the signaling domain initiated by \cite{spence1973job}, theories of equilibrium failure typically focus on sender-side incentives. High types may choose not to signal due to countersignaling motives \citep{feltovich2002too}, or markets may unravel through dynamic waiting games \citep{daley2014waiting}. Recent advances such as \cite{cisternas2025signalling} examine signaling with private monitoring, where senders do not observe the signals their actions generate. The common thread throughout this literature is that separation fails because senders' signaling incentives break down.

We introduce a receiver-side mechanism for equilibrium disruption. When buyers explore primarily for catalytic reasons, their exploration decisions become insensitive to quality signals. This destroys the feedback loop that sustains separating equilibria: high types cannot benefit from costly signaling because they are not explored preferentially. The mechanism differs from information cascades \citep{bikhchandani1992theory, bikhchandani2024information} where agents ignore private signals to follow public actions. Here, agents use their private information but explore indiscriminately to resolve status quo uncertainty.

The welfare implications diverge from standard signaling models. While signaling typically involves pure deadweight loss from a social perspective, catalytic exploration can be socially valuable by helping agents better allocate to their comparative advantage. However, the private and social values of exploration diverge because individuals do not internalize the costs imposed on those they explore. This "catalytic externality" represents a new form of market failure requiring targeted intervention, connecting to the recent work of \cite{best2025antisocial} on antisocial learning where information acquisition by some agents can harm social welfare.

\subsection{Empirical Relevance and Applications}

Three empirical patterns motivate our theoretical development. First, in online dating markets, users frequently browse profiles and initiate contact with potential partners they never meet, with contact-to-meeting conversion rates below 5\% \citep{rosenfeld2019disintermediating}. Second, in labor markets, \cite{marinescu2020opening} document that job seekers apply to many positions for which they are overqualified, with application-to-acceptance rates suggesting extensive exploration without intent to accept. Third, in consumer markets, shopping cart abandonment rates exceed 70\% despite consumers spending significant time researching products \citep{baymard2023statistics}.

These patterns share two features that standard theories struggle to explain: (i) exploration is costly and extensive, yet (ii) switching rates are negligible. While standard motives—such as searching for strictly superior alternatives—undoubtedly drive a portion of this activity, our theory rationalizes both features through the catalytic mechanism. In this framework, agents explore extensively not necessarily to find better alternatives, but to calibrate their assessment of current options. Consequently, persistent low switching rates need not imply severe search frictions or market power, but can instead emerge as the natural outcome when exploration serves primarily catalytic purposes.

The theory generates testable predictions that distinguish it from alternative explanations. Most directly, exploration intensity is strictly increasing in status quo uncertainty even when the variance of alternative quality is negligible—a prediction distinct from standard search models driven solely by external volatility. In panel data, individuals experiencing exogenous shocks to status quo uncertainty (job restructuring, health changes, relationship transitions) should exhibit increased exploration without increased switching. In experimental settings, providing subjects with precise information about their status quo should reduce exploration more than providing information about alternatives.

\subsection{Organization}

Section 2 develops the baseline model and establishes the option value decomposition. Section 3 analyzes strategic interactions when quality is endogenously signaled. Section 4 examines optimal information acquisition with endogenous precision choice. Section 5 addresses welfare implications and optimal policy. Section 6 presents empirical implications and identification strategies. Section 7 explores extensions including multi-dimensional uncertainty, dynamic learning, and network effects. Section 8 concludes. All proofs appear in the Appendix.

\section{The Baseline Model}

\subsection{Setup and Primitives}

Consider a risk-neutral Decision Maker (DM) who chooses between two mutually exclusive options: a status quo option ($O_0$) yielding payoff $u_0 = \mu_0 + \epsilon$, and a challenger option ($O_1$) yielding payoff $u_1 = \theta$. The status quo's baseline value $\mu_0$ is common knowledge. The component $\epsilon$ represents the idiosyncratic match value or subjective alignment. Crucially, $\epsilon$ is initially unknown to the DM.

To derive sharp analytical results, we impose the following distributional structure:
\begin{assumption}[Normality]\label{ass:normal}
    The match value and challenger quality follow independent normal distributions:
    \begin{equation}
        \epsilon \sim \mathcal{N}(0, \sigma_\epsilon^2) \quad \text{and} \quad \theta \sim \mathcal{N}(\mu_1, \sigma_\theta^2).
    \end{equation}
\end{assumption}

The challenger is assumed to be ex-ante inferior:
\begin{assumption}[Inferior Challenger]\label{ass:inferior}
    The challenger is inferior in expectation: $\mu_1 < \mu_0$.
\end{assumption}
This ensures a risk-neutral DM would never switch without additional information.

\paragraph{Information and Timing.}
The decision unfolds in two stages. First, the DM decides whether to explore the challenger at a sunk cost $c > 0$. Second, outcomes are determined: if the DM rejects exploration, she chooses based on priors, receiving $\max\{\mu_0, \mu_1\} = \mu_0$; if she explores, she observes the exact realizations of $(\epsilon, \theta)$ and receives $\max\{\mu_0 + \epsilon, \theta\}$.

\paragraph{Exploration Condition.}
The DM explores if and only if the expected gain exceeds the cost:
\begin{equation} \label{eq:condition}
    \E[\max\{\mu_0 + \epsilon, \theta\}] - \mu_0 \geq c
\end{equation}
We refer to the left-hand side of \eqref{eq:condition} as the Total Option Value (OV). For consistency with later welfare and dynamic comparisons, we introduce a discount factor $\delta \in (0,1]$ scaling the option value. Thus, the general exploration condition is $\delta \cdot \text{OV} \geq c$. In this static baseline, we normalize $\delta=1$, but we retain $\delta$ in the notation for subsequent propositions (e.g., Theorem \ref{thm:existence}) to explicitly capture the role of time preferences.

\subsection{Analysis of the Exploration Decision}

\subsubsection{The Option Value Decomposition}

Standard search theory attributes the value of exploration to the possibility of discovering a superior alternative. We show that this is only part of the story. By resolving uncertainty about the status quo, exploration generates value even if the challenger is ultimately rejected.

\begin{theorem}[Option Value Decomposition]\label{thm:decomposition}
The total option value can be uniquely decomposed as $\text{OV} = V^{IC} + V^{ISQ}$, where:
\begin{equation}
    V^{IC} \equiv \E[\max\{\mu_0 + \epsilon, \theta\}] - \E[\max\{\mu_0 + \epsilon, \mu_1\}]
\end{equation}
is the switching value, and
\begin{equation}
    V^{ISQ} \equiv \E[\max\{\mu_0 + \epsilon, \mu_1\}] - \mu_0
\end{equation}
is the catalytic value. Both components are non-negative. Moreover, under standard regularity conditions (e.g., full support), $V^{ISQ} > 0$ strictly whenever $\sigma_\epsilon^2 > 0$.
\end{theorem}

\begin{proof}
See Appendix A.1.
\end{proof}

The decomposition reveals that exploration can have positive value even when the probability of switching is arbitrarily small. The catalytic value $V^{ISQ}$ captures the reflexive gain: by observing the challenger (even a deterministic one), the agent creates a comparative context to realize $\epsilon$. The non-negativity follows from Jensen's inequality and the convexity of the max function.

\subsubsection{Characterization of Catalytic Value}

To derive sharp comparative statics, we assume normality. Let $\Delta \equiv \mu_0 - \mu_1 > 0$ denote the ex-ante quality gap.

\begin{proposition}[Catalytic Value: Exact Formula and Limits]\label{prop:catalytic}
Under Assumption \ref{ass:normal}, the catalytic value is given by:
\begin{equation}
    V^{ISQ} = -\Delta \Phi\left(\frac{-\Delta}{\sigma_\epsilon}\right) + \sigma_\epsilon \phi\left(\frac{-\Delta}{\sigma_\epsilon}\right)
\end{equation}
where $\phi(\cdot)$ and $\Phi(\cdot)$ are the standard normal PDF and CDF. The catalytic value satisfies:
(1) The marginal effect of uncertainty is bounded, $\frac{\partial V^{ISQ}}{\partial\sigma_\epsilon} = \phi\left(\frac{-\Delta}{\sigma_\epsilon}\right) \in (0, \frac{1}{\sqrt{2\pi}}]$; and
(2) In the high uncertainty limit ($\sigma_\epsilon \to \infty$), the value scales linearly as $V^{ISQ} \sim \frac{\sigma_\epsilon}{\sqrt{2\pi}} - \frac{\Delta}{2}$.
\end{proposition}

\begin{proof}
See Appendix A.2.
\end{proof}

This proposition establishes that catalytic value is strictly increasing in status quo uncertainty but with diminishing marginal returns (bounded by the peak of the normal density). In the high-uncertainty limit, the value scales linearly with $\sigma_\epsilon$, suggesting that for sufficiently opaque status quos, exploration becomes optimal regardless of cost.

\subsubsection{Existence of Predominantly Catalytic Exploration}

A central insight of our theory is that rational agents may explore with probability one even when the probability of switching approaches zero.

\begin{theorem}[Decoupling of Exploration and Switching]\label{thm:existence}
There exists a parameter region where exploration is optimal with probability approaching one, while the probability of switching remains strictly bounded. Specifically, fixing the quality gap $\Delta > 0$, as status quo uncertainty increases ($\sigma_{\epsilon} \rightarrow \infty$):
(i) The optimal exploration probability converges to one ($\Pr(\delta V^{ISQ} > c) \rightarrow 1$);
(ii) The switching probability approaches a strictly bounded limit ($\Pr(\theta > \mu_{0}+\epsilon) \rightarrow 0.5$), implying that for any finite $\sigma_{\epsilon}$, switching is never guaranteed despite certain exploration.
Thus, in the high-uncertainty limit, exploration is driven predominantly by catalytic motives ($V^{ISQ} \gg V^{IC}$), creating a regime of ``exploration without switching.''
\end{theorem}

\begin{proof}
See Appendix A.6.
\end{proof}

This theorem rationalizes the "window shopping" puzzle. When $\sigma_\epsilon$ is large, the option value of resolving internal uncertainty ($V^{ISQ}$) exceeds the cost $c$, triggering exploration. However, because the challenger is inferior ($\mu_1 < \mu_0$), the realization of $\epsilon$ is likely to confirm the status quo's superiority (since $\Delta$ is large relative to $\sigma_\theta$), resulting in a low switching rate.

\subsection{Comparative Statics and Properties}

\subsubsection{How Catalytic Value Varies with Uncertainty}

The relationship between uncertainty and catalytic value is central to understanding when catalytic exploration occurs.

\begin{proposition}[Comparative Statics of Catalytic Value]\label{prop:comparative}
Under Assumption \ref{ass:normal} with $\mu_1 < \mu_0$, the catalytic value $V^{ISQ}$ has the following properties with respect to $\sigma_\epsilon$:
(1) Monotonicity: $\frac{\partial V^{ISQ}}{\partial \sigma_\epsilon} > 0$ for all $\sigma_\epsilon > 0$;
(2) Bounded Marginal Effect: $0 < \frac{\partial V^{ISQ}}{\partial \sigma_\epsilon} < \frac{1}{\sqrt{2\pi}}$, and $\lim_{\sigma_\epsilon \to \infty} \frac{\partial V^{ISQ}}{\partial \sigma_\epsilon} = \frac{1}{\sqrt{2\pi}}$; and
(3) Cross Effect: $\frac{\partial^2 V^{ISQ}}{\partial \sigma_\epsilon \partial \mu_1} > 0$.
\end{proposition}

\begin{proof}
See Appendix A.3.
\end{proof}

These results imply that uncertainty and challenger quality are complements in generating catalytic value: a better (though still inferior) challenger makes status-quo uncertainty more informative and hence more valuable.

\subsubsection{Graphical Illustration}

To build intuition, we illustrate the key relationships graphically.

\begin{figure}[htbp]
\centering
\begin{tikzpicture}[scale=1.2]
\begin{axis}[
    xlabel={Status Quo Uncertainty $\sigma_\epsilon$},
    ylabel={Value},
    xmin=0, xmax=10,
    ymin=0, ymax=5,
    legend pos=north west,
    grid=major,
    width=12cm,
    height=8cm
]
\addplot[
    domain=0:10,
    samples=100,
    thick,
    color=blue
] {0.4*x - 0.5 + 0.2*exp(-x/4)};
\addlegendentry{Total Option Value (OV)}

\addplot[
    domain=0.1:10,
    samples=100,
    thick,
    color=red,
    dashed
] {0.4*x - 1 + 1/(x+0.5)};
\addlegendentry{Catalytic Value ($V^{ISQ}$)}

\addplot[
    domain=0:10,
    samples=100,
    thick,
    color=green,
    dotted
] {0.5 + 0.2*exp(-x/4)};
\addlegendentry{Switching Value ($V^{IC}$)}

\addplot[
    domain=0:10,
    samples=2,
    thick,
    color=black,
    dashed
] {1.5};
\addlegendentry{Cost Threshold ($c/\delta$)}

\draw[thick, gray, dashed] (axis cs:3.33,0) -- (axis cs:3.33,5);
\node at (axis cs:3.33,4.5) [right] {$\sigma_\epsilon^*$};
\end{axis}
\end{tikzpicture}
\caption{Decomposition of Option Value}
\label{fig:decomposition}
\end{figure}
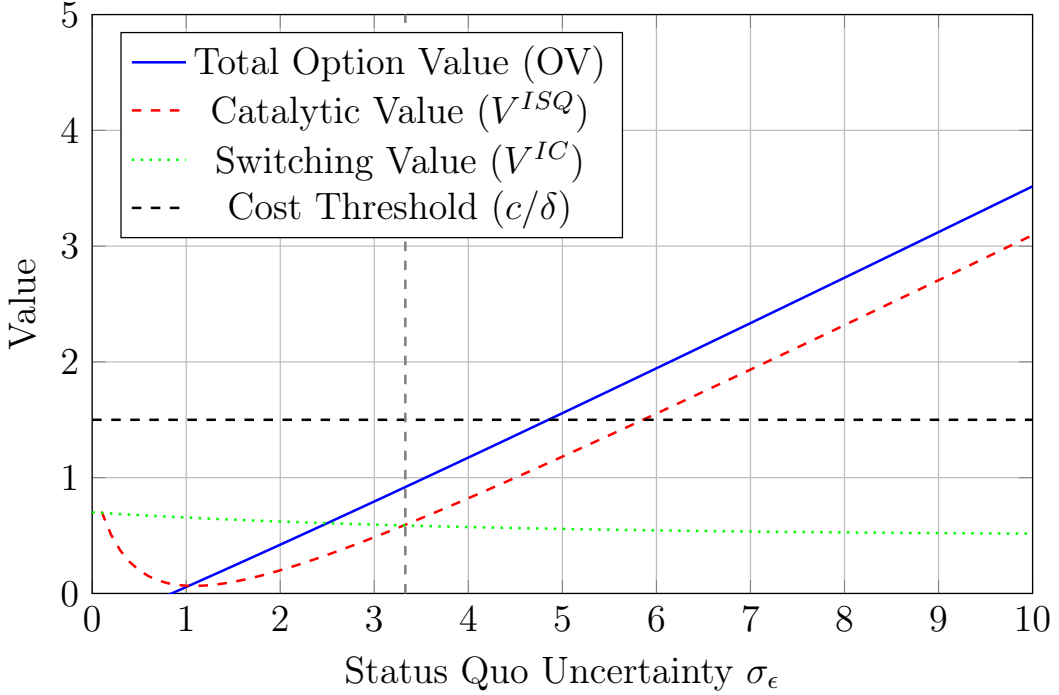

Figure \ref{fig:decomposition} illustrates that for high levels of uncertainty, exploration is justified ($\text{OV} > c$) but driven predominantly by catalytic motives ($V^{ISQ} \gg V^{IC}$).

\subsubsection{The Role of Risk Aversion}

So far, we have assumed risk neutrality. How does risk aversion affect catalytic exploration?

\begin{proposition}[Risk Aversion and Catalytic Value]\label{prop:risk}
Consider a decision-maker with CARA utility $u(x) = -e^{-\gamma x}$ where $\gamma > 0$. Let $V^{ISQ}_{\text{CARA}}(\gamma)$ denote the catalytic value. We find that risk aversion increases the value of resolving status-quo uncertainty ($V^{ISQ}_{\text{CARA}}(\gamma) > V^{ISQ}_{\text{RN}}$). This amplification occurs through an insurance effect: exploration allows the DM to avoid committing to the status quo when $\epsilon$ realizes negatively. The magnitude of this effect is increasing in $\gamma$ and $\sigma_\epsilon$.
\end{proposition}

\begin{proof}
See Appendix A.4.
\end{proof}

\paragraph{Remark.}
While a closed-form solution for $V^{ISQ}_{\text{CARA}}$ is intractable, numerical analysis in Appendix H confirms that risk-averse agents value the insurance aspect of catalytic exploration beyond its pure information value.

\subsection{Dynamic Extension: Repeated Exploration Opportunities}

Consider an infinite-horizon extension where challenger options arrive according to a Poisson process with rate $\lambda$.

\begin{proposition}[Stationary Exploration Policy]\label{prop:dynamic}
In the infinite-horizon model with discount rate $r$, the optimal policy follows a threshold rule: explore if and only if $\sigma_\epsilon > \bar{\sigma}_\epsilon$. The value function $V(\sigma_\epsilon)$ satisfies the HJB equation:
\begin{equation}
    rV(\sigma_\epsilon) = \mu_0 + \lambda \max\{-c + V^{ISQ}(\sigma_\epsilon) + [V(0) - V(\sigma_\epsilon)], 0\}
\end{equation}
\noindent \textit{Remark:} This formulation focuses on the catalytic-dominant regime where switching value is negligible ($V^{IC} \approx 0$), allowing the total option value to be approximated by $V^{ISQ}$. In the high-uncertainty region, the threshold approximately satisfies $V^{ISQ}(\bar{\sigma}_\epsilon) = c$, yielding $\bar{\sigma}_\epsilon \approx c \sqrt{2\pi}$. Catalytic exploration persists in steady state if uncertainty regenerates over time.
\end{proposition}

\begin{proof}
See Appendix A.5.
\end{proof}

This extension confirms that catalytic exploration is not merely a one-shot phenomenon but persists in dynamic settings where uncertainty is recurrent.

\section{Strategic Interactions: Signaling with Catalytic Motives}

\subsection{Model Setup with Strategic Challengers}

We now extend the baseline model to a strategic environment where challengers are not passive probability distributions but active agents who can signal their quality. This extension is motivated by labor markets (education signaling), dating markets (courtship displays), and product markets (advertising), where ``exploration'' often follows a costly signal sent by the challenger.

\subsubsection{Players and Timing}

Consider a game between a Decision Maker (DM) and a Challenger (Sender). The Challenger's type $\theta \in \{\theta_L, \theta_H\}$ is private information, with $\theta_H > \theta_L$. The prior probability of the high type is $\Pr(\theta = \theta_H) = p \in (0, 1)$. Consistent with Assumption \ref{ass:inferior}, the challenger is ex-ante inferior, $\mu_1 \equiv p\theta_H + (1-p)\theta_L < \mu_0$.

The interaction unfolds in three stages. First, in the Signaling Stage, Nature determines $\theta$, and the Challenger chooses an observable signal (effort) $e \in [0, \infty)$ at cost $\psi(e, \theta)$. Next, in the Exploration Stage, the DM observes $e$, forms posterior belief $\hat{p}(e) = \Pr(\theta_H|e)$, and decides whether to explore at cost $c$ or reject. Finally, in the Selection Stage, outcomes are realized. If the DM rejects, she receives $\mu_0$ and the Challenger receives a normalized reservation utility of 0 minus signaling costs. If she explores, she learns the true realization of $\theta$ and the status quo match value $\epsilon$, choosing the best option. The Challenger receives a reward $W > 0$ if explored. We assume the standard single-crossing property holds: the marginal cost of signaling is strictly lower for the high type. Future payoffs are discounted by $\delta \in (0, 1]$.

\subsection{Equilibrium Analysis}

We characterize Perfect Bayesian Equilibria (PBE). A PBE consists of the Challenger's signaling strategy $e^*(\theta)$, the DM's exploration rule $a^*(e)$, and consistent beliefs.

\subsubsection{The Catalytic Disruption of Signaling}

In standard signaling models (e.g., \citealp{spence1973job}), separation is sustained because the receiver responds favorably only to high signals. Catalytic exploration fundamentally alters this logic. Even if the DM is certain the challenger is a low type ($\theta_L$), she may still explore if the catalytic value of resolving status quo uncertainty is sufficiently high.

\begin{theorem}[Catalytic Disruption of Separating Equilibrium]\label{thm:disruption}
Let $V^{IC}(\theta)$ denote the switching value conditional on type $\theta$. Define the critical catalytic threshold $\bar{V}^{ISQ} \equiv \frac{c}{\delta} - V^{IC}(\theta_L)$.
\begin{enumerate}
    \item If $V^{ISQ} < \bar{V}^{ISQ}$, the unique outcome satisfying the Intuitive Criterion is the least-cost separating equilibrium, where $\theta_L$ chooses $e_L=0$ and $\theta_H$ chooses $e_H^* > 0$.
    \item If $V^{ISQ} > \bar{V}^{ISQ}$, no separating equilibrium exists. The unique equilibrium is pooling at zero effort ($e_L = e_H = 0$).
\end{enumerate}
At the threshold $V^{ISQ} = \bar{V}^{ISQ}$, the equilibrium effort $e^*(\theta_H)$ drops discontinuously from a positive level to zero.
\end{theorem}

\begin{proof}
See Appendix B.1.
\end{proof}

The threshold $\bar{V}^{ISQ}$ represents the point where catalytic motives become so strong that the DM finds it optimal to explore even a known low-quality challenger. In the ``Collapse Region,'' the DM's exploration strategy becomes insensitive to the signal: she explores everyone indiscriminately to resolve her own status quo uncertainty. Since the high type gets explored regardless of effort, she has no incentive to incur the signaling cost $e_H^*$. The separating equilibrium collapses not because high types cannot signal, but because receivers do not need the signal to justify exploration.

\subsubsection{Welfare Implications}

The collapse of signaling has ambiguous welfare effects, creating a trade-off between information transmission and resource conservation.

\begin{proposition}[Welfare Effects of Disruption]\label{prop:welfare_signal}
Comparing the pooling equilibrium (under high catalytic value) to the separating equilibrium, we find divergent welfare impacts. For challengers, high types are worse off due to the loss of distinction, while low types gain from increased exploration chances. For total social welfare, the disruption is beneficial if and only if the saved signaling costs exceed the deadweight loss from exploring low types: $p \cdot \psi(e_H^*, \theta_H) > (1-p) \cdot ( c - \delta \cdot \text{OV}(\theta_L) )$.
\end{proposition}

\begin{proof}
See Appendix B.2.
\end{proof}

This result highlights a novel ``second-best'' property of catalytic exploration: by making buyers ``too curious'' (exploring low types), it eliminates the wasteful signaling rat race for high types.

\subsection{Comparative Statics and Extensions}

In this section, we analyze the determinants of the signaling breakdown. We begin by examining the comparative statics of the binary type model established in Theorem 3. Subsequently, to characterize the gradual erosion of the separating region, we will extend the type space to a continuum.

\subsubsection{Determinants of Collapse (Binary Regime)}

\begin{proposition}[Threshold Comparative Statics]\label{prop:threshold}
The signaling equilibrium is more likely to collapse (i.e., the critical threshold $\bar{V}^{ISQ}$ decreases) under the following conditions:
\begin{itemize}
    \item \textbf{Lower search costs} ($\partial \bar{V}^{ISQ}/\partial c > 0$): Cheaper exploration directly lowers the barrier to indiscriminately testing low types.
    \item \textbf{Greater patience} ($\partial \bar{V}^{ISQ}/\partial \delta < 0$): A higher discount factor increases the present value of the exploration payoff.
    \item \textbf{Higher quality of low types} ($\partial \bar{V}^{ISQ}/\partial \theta_L < 0$): As $\theta_L$ improves, the \textit{switching value} of exploring them increases ($V^{IC}(\theta_L)$ rises), making the DM more willing to explore even the worst candidates.
\end{itemize}
\end{proposition}

\begin{proof}
See Appendix B.3.
\end{proof}

This suggests that improvements in information technology (lower $c$) can paradoxically reduce market transparency by destroying the incentives for quality signaling.

\subsubsection{Continuous Types and Partial Pooling}

To rigorously characterize how the separating equilibrium unravels, we extend the type space to a continuum $\theta \sim G(\cdot)$ on $[\underline{\theta}, \bar{\theta}]$. In this setting, the collapse of signaling manifests not as a sudden jump, but as a gradual erosion of the separating region.

\begin{theorem}[Equilibrium Regimes]\label{thm:riley}
In the continuous-type model, the equilibrium structure depends on the magnitude of catalytic value $V^{ISQ}$:
\begin{enumerate}
    \item \textbf{Full Separation:} If $V^{ISQ}$ is low ($V^{ISQ} < c/\delta - V^{IC}(\underline{\theta})$), a unique Riley separating equilibrium exists.
    \item \textbf{Complete Collapse:} If $V^{ISQ}$ is sufficiently high ($V^{ISQ} > c/\delta - V^{IC}(\bar{\theta})$), the unique equilibrium is complete pooling at zero effort where all types are explored.
    \item \textbf{Partial Separation with Rejection:} For intermediate catalytic value, the equilibrium involves pooling at the bottom (who are rejected) and separation at the top (who are explored).
\end{enumerate}
\end{theorem}

\begin{proof}
See Appendix B.4.
\end{proof}
We characterize the intermediate regime specifically, as it represents the erosion of market transparency.

\begin{proposition}[Structure of Partial Pooling]\label{prop:partial}
In the partial pooling regime, there exists a cutoff type $\hat{\theta} \in (\underline{\theta}, \bar{\theta})$ such that types $\theta < \hat{\theta}$ pool at $e^*(\theta) = 0$ and are rejected by the DM (as the option value is insufficient), while types $\theta \ge \hat{\theta}$ separate with strictly positive effort and are explored. The exclusion region shrinks as status quo uncertainty increases ($\partial \hat{\theta} / \partial V^{ISQ} < 0$).
\end{proposition}

\begin{proof}
See Appendix B.5.
\end{proof}

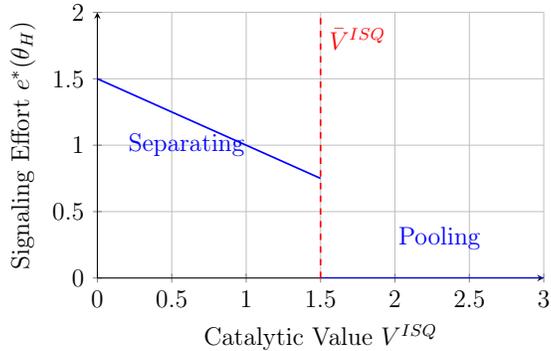
\begin{figure}[htbp]
\centering
\begin{tikzpicture}[scale=0.8] 
\begin{axis}[
    xlabel={Catalytic Value $V^{ISQ}$},
    ylabel={Signaling Effort $e^*(\theta_H)$},
    xmin=0, xmax=3,
    ymin=0, ymax=2,
    grid=major,
    width=9cm,   
    height=6cm,  
    axis lines=left
]
\addplot[domain=0:1.5, samples=100, thick, blue] {1.5 - 0.5*x};
\addplot[domain=1.5:3, samples=2, thick, blue] {0};

\draw[thick, red, dashed] (axis cs:1.5,0) -- (axis cs:1.5,2);
\node at (axis cs:1.5,1.8) [right, red] {$\bar{V}^{ISQ}$};

\node at (axis cs:0.6,1.0) [blue] {Separating};
\node at (axis cs:2.3,0.3) [blue] {Pooling};
\end{axis}
\end{tikzpicture}
\caption{Discontinuous Collapse of Signaling Effort (Binary Regime)}
\label{fig:signaling}
\end{figure}

\section{Endogenous Information Acquisition}

\subsection{The Information Design Problem}

We now relax the assumption of perfect learning. In many contexts—such as evaluating a complex job offer or assessing the long-term compatibility of a partner—information is costly and signals are noisy. We endogenize the information structure by allowing the DM to choose the quality of her signals before exploring.

\subsubsection{Signal Structure and Costs}

If the DM chooses to explore, she receives unbiased Gaussian signals about the status quo ($O_0$) and the challenger ($O_1$):
\begin{align}
    s_\epsilon &= \epsilon + \eta, \quad \eta \sim \mathcal{N}(0, \tau_\eta^2) \\
    s_\theta &= \theta + \xi, \quad \xi \sim \mathcal{N}(0, \tau_\xi^2)
\end{align}
where $\tau_\eta, \tau_\xi \in (0, \infty)$ represent the \textbf{standard deviation of the noise}. A lower $\tau$ corresponds to higher signal precision. The noise terms are independent of the fundamentals.

The DM chooses $(\tau_\eta, \tau_\xi)$ to minimize the sum of exploration costs $c$ and information acquisition costs $C(\tau_\eta, \tau_\xi)$. We assume a cost function that is strictly increasing in precision (or equivalently, decreasing in noise). Given our parameterization by noise scale $\tau$, we adopt the specific functional form:
\begin{equation}
    C(\tau_\eta, \tau_\xi) = \frac{\kappa_\eta}{\tau_\eta^2} + \frac{\kappa_\xi}{\tau_\xi^2}
\end{equation}
where $\kappa_\eta, \kappa_\xi > 0$ are technology parameters scaling the cost of information. This specification satisfies the Inada conditions: the marginal cost of noise reduction becomes infinite as $\tau \to 0$, preventing the acquisition of perfect information.

\subsubsection{The Optimization Problem}

The DM's problem is to select noise levels to maximize the net option value. Let $\Sigma_\epsilon(\tau_\eta)$ denote the posterior variance (which is increasing in $\tau_\eta$). The DM solves:
\begin{equation}
    \max_{\tau_\eta, \tau_\xi > 0} \delta \cdot \text{OV}(\tau_\eta, \tau_\xi) - \left( c + \frac{\kappa_\eta}{\tau_\eta^2} + \frac{\kappa_\xi}{\tau_\xi^2} \right)
\end{equation}
where the total option value $\text{OV}(\tau_\eta, \tau_\xi)$ is the expected gain from making the choice based on posterior beliefs.

This setup parallels the rational inattention framework in dynamic discrete choice \citep{miao2024dynamic}, where agents optimally allocate attention to resolve the most decision-relevant uncertainty. However, unlike standard models where attention flows to the most uncertain options, our catalytic mechanism generates a distinct allocation pattern.

\subsection{Optimal Information Acquisition}

\subsubsection{Asymmetric Attention}

The first-order conditions reveal a fundamental asymmetry in how agents learn about the status quo versus the challenger.

\begin{theorem}[Optimal Information Design: Asymmetric Attention]\label{thm:info_design}
The optimal noise levels $(\tau_\eta^*, \tau_\xi^*)$ exhibit a distinct focus on the status quo. Specifically, if the challenger is sufficiently inferior ($\mu_1 \ll \mu_0$), the DM acquires significantly more precise information about the status quo than about the challenger ($\tau_\eta^* \ll \tau_\xi^*$). Furthermore, as the catalytic motive becomes dominant ($V^{ISQ}/V^{IC} \to \infty$), the agent ceases learning about the challenger while maintaining high precision on the status quo:
\begin{equation}
    \lim_{\Delta \to \infty} \tau_\xi^* = \infty, \quad \text{while } \tau_\eta^* \to \bar{\tau} < \infty.
\end{equation}
\end{theorem}

\begin{proof}
See Appendix C.1.
\end{proof}

\paragraph{Intuition.}
Information about the status quo ($\epsilon$) has global value: learning $\epsilon$ helps the DM evaluate her current position regardless of the specific challenger. In contrast, information about the challenger ($\theta$) has only local value: it is useful only if $\theta$ turns out to be high enough to induce a switch. When the challenger is ex-ante inferior, the probability of $\theta$ being pivotal is low, making information about it almost worthless ($V^{IC} \approx 0$). Thus, rational agents focus entirely on "self-discovery" ($\epsilon$) while ignoring the details of the alternative they are using as a benchmark.

\subsubsection{Comparative Statics}

\begin{proposition}[Information Complementarities]\label{prop:info_comp}
The optimal information structure exhibits strategic interaction between sources.
First, we find strategic complementarity in status quo uncertainty: higher prior uncertainty $\sigma_\epsilon$ induces higher precision acquisition (lower $\tau_\eta^*$) when $\sigma_\epsilon$ is low.
Second, under a binding attention constraint, we identify source substitution: an increase in the cost of learning about the challenger ($\kappa_\xi$) leads to more precise learning about the status quo ($\partial \tau_\eta^*/\partial \kappa_\xi < 0$).
\end{proposition}

\begin{proof}
See Appendix C.2.
\end{proof}

\subsection{The Information Technology Paradox}

Standard search theory suggests that reducing information costs improves welfare by facilitating better matches. Our model identifies a perverse outcome driven by catalytic externalities.

\begin{theorem}[Information Technology Paradox]\label{thm:it_paradox}
Consider a technological improvement that reduces the cost of acquiring information about the status quo ($\kappa_\eta' < \kappa_\eta$). This leads to a divergence between individual incentives and social welfare. While the DM optimally increases exploration frequency and acquires more precise signals, total social welfare decreases if the marginal catalytic externality exceeds the private value of refinement:
\begin{equation}
    \frac{\partial \mathcal{W}}{\partial \kappa_\eta} > 0 \iff \text{Marginal Externality} > \text{Marginal Private Option Value}
\end{equation}
\end{theorem}

\begin{proof}
See Appendix C.3.
\end{proof}

\subsubsection{}{Interpretation.}
Cheaper information (e.g., Tinder, LinkedIn, Zillow) lowers the threshold for catalytic exploration. Agents begin to explore alternatives not because they are genuine candidates, but because it is cheap to use them as mirrors for self-evaluation. While individually rational, this floods the market with "noise traders" (catalytic explorers) who impose congestion and screening costs on challengers (firms, potential partners), potentially reducing aggregate efficiency. This mechanism provides a theoretical basis for the "paradox of choice" in digital markets.

\section{Welfare Analysis and Policy Implications}

\subsection{Social Efficiency and Market Failures}

We now turn to the normative implications of the theory. Does the equilibrium level of exploration maximize social welfare? The answer depends on whether the resolution of status quo uncertainty generates social value beyond the DM's private utility, and whether the costs of exploration are fully internalized.

\subsubsection{The Social Planner's Problem}

Consider a utilitarian social planner who maximizes the sum of the DM's and the Challenger's surplus, net of resource costs. Let $\mathcal{E} > 0$ denote the uncompensated cost imposed on the Challenger during exploration (e.g., time spent interviewing, processing applications, or sustaining a courtship). Conversely, let $\mathcal{S} \ge 0$ denote potential positive information spillovers (e.g., market thickness benefits or social learning).

The social net value of exploration is:
\begin{equation}
    \Delta W_{soc} = \underbrace{\delta (V^{IC} + V^{ISQ})}_{\text{DM's Gain}} - \underbrace{c}_{\text{Resource Cost}} - \underbrace{\delta (\mathcal{E} - \mathcal{S})}_{\text{Net Externality}}
\end{equation}
The DM, however, explores whenever $\delta (V^{IC} + V^{ISQ}) \ge c$.

\begin{theorem}[Inefficiency of Equilibrium]\label{thm:inefficiency}
The equilibrium exploration level diverges from the social optimum whenever the net externality is non-zero. Specifically, excessive exploration occurs when catalytic externalities dominate ($\mathcal{E} > \mathcal{S}$), particularly in the ``predominantly catalytic'' regime where private motives ($V^{ISQ}$) ignore social costs. Conversely, under-exploration may occur if positive spillovers are large ($\mathcal{S} > \mathcal{E}$), but catalytic value is too low to trigger private action.
\end{theorem}

\begin{proof}
See Appendix D.1.
\end{proof}

\subsubsection{Sources of Inefficiency}

Our model identifies three distinct channels of market failure extending beyond standard search frictions. First, and most central to our theory, is the Catalytic Externality. Individuals utilize challengers as informational instruments to calibrate their internal preferences ($V^{ISQ}$), imposing inspection costs on challengers without a commensurate probability of transaction. This mechanism closely parallels the ``antisocial learning'' identified by \cite{best2025antisocial}, where private learning incentives yield negative social payoffs. Second, Information Spillovers may arise in networked settings where one agent's exploration reveals quality information to others. Since the DM ignores the value of generated public information ($V^{IC}$ for others), this can lead to under-exploration. Finally, Congestion Effects occur when exploration resources are scarce; catalytic explorers—behaving effectively as ``noise traders''—crowd out agents with genuine switching intent.

\subsection{Optimal Policy Interventions}

\subsubsection{Pigouvian Corrections}

When exploration exerts net negative externalities, the canonical remedy is a Pigouvian tax.

\begin{proposition}[Optimal Tax Policy]\label{prop:tax}
The socially optimal intervention involves a two-part instrument.
First, an Exploration Tax $t = \delta \mathcal{E}$ is levied on the act of exploration (e.g., application fees, viewing fees) to internalize the catalytic externality.
Second, a Matching Subsidy $\sigma = \mathcal{S}/P_s$ is provided conditional on switching to encourage productive exploration that generates positive spillovers.
Net taxation is optimal when the probability of switching is low, i.e., when catalytic motives dominate ($V^{ISQ} \gg V^{IC}$).
\end{proposition}

\begin{proof}
See Appendix D.2.
\end{proof}

\subsubsection{Information Disclosure Regulation}

An alternative to pricing exploration is to resolve status quo uncertainty directly through mandatory information disclosure (e.g., salary transparency laws, product labeling standards).

\begin{theorem}[Disclosure Policy]\label{thm:disclosure}
A policy that publicly reveals $\epsilon$ (eliminating status quo uncertainty) has ambiguous welfare effects.
On the one hand, it yields a Benefit by eliminating wasteful catalytic exploration ($V^{ISQ} \to 0$), thereby saving costs $c$ and externality $\mathcal{E}$.
On the other hand, it imposes a Cost by potentially destroying beneficial signaling equilibria (as shown in Section 3), forcing high-quality types to incur signaling costs again or leading to inefficient pooling.
\end{theorem}

\begin{proof}
See Appendix D.3.
\end{proof}

\subsection{Market Design Implications}

The theory offers a unified lens for understanding market design features that seem puzzling under standard search theory.

\paragraph{Labor Markets (Screening Mechanisms).}
\cite{marinescu2020opening} document that job seekers send many applications with low intent to accept. Our theory suggests this is partly to gauge market value. Optimal design requires separating validation seekers from genuine candidates. Application fees can efficiently screen out catalytic applicants who explore only to validate current employment terms. Furthermore, enhancing information provision through career counseling or transparent salary benchmarks can resolve $\epsilon$ directly, reducing the need for catalytic interviews.

\paragraph{Dating Markets (Friction Design).}
Platform features can be understood as regulating catalytic exploration. By increasing the shadow cost of exploration, swipe limits force users to reserve exploration for candidates with high potential for adoption ($V^{IC}$), rather than for validation. Similarly, commitment devices such as ``Exclusive Mode'' prevent ongoing catalytic exploration that undermines relationship stability.

\paragraph{Innovation Policy (R\&D Incentives).}
Firms often conduct R\&D on technologies they do not commercialize to benchmark legacy systems. If this ``defensive research'' crowds out viable entrants, subsidies should be targeted toward implementation milestones rather than pure exploration expenditure.

\section{Empirical Implications and Identification}

\subsection{Testable Predictions}

The theory generates sharp testable predictions that distinguish catalytic exploration from standard search models. We organize these into cross-sectional patterns (comparing agents with different uncertainty levels) and dynamic patterns (changes over time).

\subsubsection{Cross-Sectional Predictions}

\begin{proposition}[The Exploration-Switching Paradox]\label{prop:empirical}
Conditional on the challenger being ex-ante inferior ($\mu_1 < \mu_0$), the model predicts:
\begin{enumerate}
    \item \textbf{Decoupling Effect:} Across agents with varying status quo uncertainty $\sigma_\epsilon$, exploration intensity increases rapidly to saturation (100\%), while switching rates increase sluggishly and remain bounded. This generates a widening ``validation gap'' in high-uncertainty environments.
    
    \item \textbf{U-Shaped Duration Dependence:} If status quo uncertainty is high early in a relationship (learning match quality), low in the middle (stable match), and high again later (preference drift or market changes), then:
    \begin{equation}
        \text{Exploration Intensity} = f(\text{Duration}) \quad \text{is U-shaped}
    \end{equation}
    
    \item \textbf{Non-Monotonic Signaling:} In markets with endogenous signaling, the intensity of quality signaling responds non-monotonically to market competitiveness (proxy for average $\sigma_\epsilon$).
    Specifically, when catalytic motives dominate ($\sigma_\epsilon$ is very large), signaling collapses:
    \begin{equation}
        \frac{\partial^2 e^*}{\partial \text{Comp} \cdot \partial \sigma_\epsilon} < 0 \quad \text{for } \sigma_\epsilon > \bar{\sigma}
    \end{equation}
\end{enumerate}
\end{proposition}

\begin{proof}
See Appendix E.1.
\end{proof}

\subsubsection{Dynamic Predictions}

\begin{proposition}[Dynamic Implications]\label{prop:dynamic_empirical}
In panel data tracking individual choices over time:
\begin{enumerate}
    \item \textbf{Validation Effect:} Successful exploration (without switching) causally increases the agent's subjective valuation of the status quo. Thus, exploration waves should precede periods of increased status quo satisfaction or reduced turnover hazard.
    \item \textbf{Technology Paradox:} Improvements in information technology (lower search cost $c$) lead to a divergence between search volume and turnover. Specifically, we observe that $\partial \text{Exploration}/\partial c < 0$, while $\partial \text{Switching}/\partial c \approx 0$ in the catalytic regime.
    \item \textbf{Reaction to Shocks:} Exogenous shocks that increase status quo uncertainty (e.g., corporate restructuring, health scares) trigger immediate exploration waves that are mostly non-consummated (high rejection rates).
\end{enumerate}
\end{proposition}

\begin{proof}
See Appendix E.2.
\end{proof}

\subsection{Identification Strategy}

To empirically isolate catalytic exploration, the econometric challenge is to distinguish it from standard search for better alternatives. We propose two complementary strategies.

\subsubsection{Instrumental Variables Approach}

Identification requires variation in status quo uncertainty ($\sigma_\epsilon$) that is orthogonal to the quality of alternatives ($\mu_1$) and the instrumental value of switching ($V^{IC}$). A promising avenue involves exploiting information shocks, such as regulatory changes regarding information disclosure (e.g., pay transparency laws).

By treating the staggered adoption of such requirements as an instrument ($Z$), we can isolate changes in $\sigma_\epsilon$. The first stage operates because disclosure directly resolves uncertainty externally. The exclusion restriction holds provided that disclosure does not alter the fundamental distribution of outside offers ($\theta$). Under this framework, the theory yields a sharp falsification test: if exploration is catalytic, transparency ($Z \uparrow$) should reduce search activity; if exploration is instrumental (standard search), transparency should have a neutral or positive effect by lowering comparison costs.

\subsubsection{Structural Estimation}

We propose a two-step structural estimation procedure to recover the model parameters from observed behavior.To bridge the theoretical binary decision with aggregate data, we assume unobserved heterogeneity in exploration costs $c \sim F_c(\cdot)$. Thus, the model-implied exploration probability is given by $P_{exp} = F_c(\delta V^{ISQ}(\sigma_\epsilon, \Delta))$. For illustration, we proceed with a simplified moment-matching approach. In the first step, the researcher estimates the reduced-form exploration probability $\hat{P}_{exp}$ and the switching probability conditional on exploration $\hat{P}_{switch}$ directly from the data. 

In the second step, the structural parameters $(\hat{\sigma}_\epsilon, \hat{\Delta})$ are recovered via the method of moments by solving the system:
\begin{align}
    \hat{P}_{exp} = F_c\left(\delta V^{ISQ}(\sigma_\epsilon, \Delta); \theta_c\right)\\
    \hat{P}_{switch} &= \Phi\left(\frac{-\Delta}{\sigma_\epsilon}\right)
\end{align}
Finally, the implied catalytic value is computed using the formula from Proposition 2:
\begin{equation}
    \hat{V}^{ISQ} = -\hat{\Delta} \Phi\left(\frac{-\hat{\Delta}}{\hat{\sigma}_\epsilon}\right) + \hat{\sigma}_\epsilon \phi\left(\frac{-\hat{\Delta}}{\hat{\sigma}_\epsilon}\right)
\end{equation}
This approach allows for the quantification of the "catalytic share" of total search activity in a given market.

\subsection{Existing Evidence}

Recent empirical findings in labor, consumer, and relationship markets align closely with the predictions of catalytic exploration.

\paragraph{Labor Markets.}
\cite{marinescu2020opening} document a phenomenon of "application spamming," where job seekers frequently apply to positions offering lower wages than their current jobs. Standard search theory struggles to explain why workers incur costs to explore strictly dominated options. Our model rationalizes this behavior as catalytic exploration: workers use these interviews not to switch, but to benchmark their current employment terms, thereby resolving uncertainty about their market value ($\sigma_\epsilon$).

\paragraph{Consumer Behavior.}
\cite{iyengar2000choice} famously document that increasing choice options can reduce purchase rates—the "paradox of choice." While typically attributed to cognitive overload, our model offers a rational information-based explanation: more options decrease the cost of acquiring comparative information ($c$), encouraging consumers to engage in extensive "window shopping" (catalytic exploration) that ultimately confirms the status quo, appearing in the data as a reduction in purchasing (switching).

\paragraph{Relationship Markets.}
\cite{rosenfeld2019disintermediating} show that the rise of online dating has drastically increased search volume (browsing, messaging) but has not led to a proportional increase in matching or relationship formation rates. This pattern—high exploration, low switching—is the signature of the catalytic regime. As search costs fall, users increasingly utilize the platform to validate their current relationship status or market desirability rather than to find new partners.

\section{Extensions and Generalizations}

\subsection{Multi-Dimensional Uncertainty}

\subsubsection{Vector of Unknown Attributes}

In reality, the status quo is rarely defined by a single attribute. Consider a status quo with multiple uncertain attributes $\boldsymbol{\epsilon} = (\epsilon_1, ..., \epsilon_n)$, where each $\epsilon_i \sim \mathcal{N}(0, \sigma^2)$ represents uncertainty along dimension $i$ (e.g., salary, location, collegiality).

\begin{proposition}[Multi-Dimensional Catalytic Value]\label{prop:multidim}
Assume payoffs are additive across $n$ independent dimensions. The total standard deviation scales with $\sqrt{n}\sigma$. Consequently, in the high-uncertainty limit, catalytic value grows as $V^{ISQ} \propto \sqrt{n}$. This growth exhibits diminishing returns: the marginal increase in catalytic value from adding an uncertainty dimension converges to zero as $n \to \infty$, implying that optimal exploration intensity is concave with respect to the complexity ($n$) of the status quo.
\end{proposition}

\begin{proof}
See Appendix F.1.
\end{proof}

\subsubsection{Selective Attention}

When the DM faces an attention constraint (total attention budget $\bar{A}$), she must choose which dimensions to investigate.

\begin{theorem}[Optimal Attention Allocation]\label{thm:attention}
Let $a_i$ denote attention allocated to dimension $i$. The optimal allocation satisfies:
\begin{equation}
    a_i^* = \bar{A} \cdot \frac{\sigma_{\epsilon_i}^2 / \tau_i^2}{\sum_{j=1}^n \sigma_{\epsilon_j}^2 / \tau_j^2}
\end{equation}
Optimally, attention is proportional to the signal-to-noise ratio of the uncertainty. The DM focuses on dimensions where she is most uncertain ($\sigma_{\epsilon_i}^2$ is high) and where information is most precise ($\tau_i^2$ is low).
\end{theorem}

\begin{proof}
See Appendix F.2.
\end{proof}

\subsection{Network Effects and Social Learning}

\subsubsection{Information Spillovers}

Consider a network of $N$ agents where one agent's exploration generates observable information for neighbors.

\begin{proposition}[Strategic Complementarity in Exploration]\label{prop:network}
In a network with information spillovers, individual exploration decisions become strategic complements when network density exceeds a threshold $\bar{\rho}$. Exploring reveals information that reduces neighbors' uncertainty, which (counter-intuitively) may encourage them to explore to resolve their own remaining residual uncertainty. This structure generates multiple equilibria—both ``all explore'' and ``none explore'' can be stable—and implies that the socially optimal outcome requires coordination, as private agents ignore the positive information externalities.
\end{proposition}

\begin{proof}
See Appendix F.3.
\end{proof}

\subsubsection{Herding and Information Cascades}

When agents observe others' exploration decisions (but not private signals), catalytic motives can generate inefficient herding.

\begin{theorem}[Catalytic Cascades]\label{thm:cascade}
With sequential exploration decisions, a ``catalytic cascade'' occurs when agents explore solely because they observe others exploring, disregarding their private signal that the challenger is inferior. These cascades are more sensitive to the magnitude of $V^{ISQ}$, as the threshold for following the herd is lower. Crucially, the welfare loss from catalytic cascades exceeds that of standard informational cascades \citep{bikhchandani1992theory} because it involves not just wrong choices, but also the negative externalities (congestion/inspection costs) imposed on challengers.
\end{theorem}

\begin{proof}
See Appendix F.4.
\end{proof}

\subsection{Dynamic Learning and Experimentation}

\subsubsection{Multi-Armed Bandit with Catalytic Value}

Consider a dynamic bandit problem where pulling an arm provides information about both that arm and the current status quo.

\begin{proposition}[Modified Gittins Index]\label{prop:gittins}
The optimal policy is characterized by a modified Gittins index:
\begin{equation}
    G_i^* = G_i^{standard} + \beta \cdot V^{ISQ}_i
\end{equation}
where $G_i^{standard}$ is the classical index for arm $i$, and the term $\beta \cdot V^{ISQ}_i$ represents the information spillover of pulling arm $i$ on the status quo's valuation. This modification justifies exploring arms with low direct payoff potential if they offer high catalytic information.
\end{proposition}

\begin{proof}
See Appendix F.5.
\end{proof}

\subsubsection{Optimal Stopping}

When should catalytic exploration cease?

\begin{theorem}[Optimal Stopping Rule]\label{thm:stopping}
In a sequential learning setting, the optimal stopping time $\tau^*$ is characterized by a threshold on the posterior uncertainty:
\begin{equation}
    \tau^* = \inf\left\{t: \sigma_{\epsilon,t} < \bar{\sigma}_\epsilon\right\}
\end{equation}
The DM stops exploring not when she finds a good option, but when she becomes sufficiently certain about the status quo. The threshold $\bar{\sigma}_\epsilon$ is structurally analogous to the static exploration boundary, adjusted for the continuation value of learning.
\end{theorem}

\begin{proof}
See Appendix F.6.
\end{proof}

\section{Conclusion}

\subsection{Summary of Contributions}

This paper develops a theory of catalytic exploration to explain a pervasive economic puzzle: why rational agents invest costly resources in exploring options they have little intention of adopting. The central insight is that exploration serves a dual informational purpose. It is not merely an instrument for discovering superior alternatives, but a mechanism for resolving endogenous uncertainty about the status quo through comparative evaluation.

Our analysis advances the literature along four dimensions. Conceptually, we introduce the notion of ``reflexive choice'' to formalize the intuition that external search generates self-knowledge within a neoclassical framework. Technically, we provide a rigorous decomposition of option value and derive a closed-form solution for the catalytic value, establishing that it scales linearly with status quo uncertainty in the limit. This asymptotic property proves that uncertainty alone can drive exploration even when the probability of switching approaches zero.

Strategically, we demonstrate that these catalytic motives can destabilize markets. When receivers explore indiscriminately to validate their own preferences, the single-crossing property fails to separate types, leading to the collapse of signaling equilibria. Finally, our welfare analysis identifies the ``catalytic externality'' as a novel source of market failure. Because agents do not internalize the inspection costs imposed on the benchmarks they explore, equilibrium outcomes are characterized by excessive exploration, justifying targeted policy interventions such as application fees or interaction limits.

\subsection{Final Remarks}

Canonical search theory typically views exploration as a means to an end—specifically, the end of replacing a current match with a better one. Our framework suggests a more nuanced view: exploration can be a means to affirm the end one already possesses.

In an era of falling search costs and abundant options, the binding constraint for decision-makers is often not the availability of alternatives, but the precision of their own preference evaluations. In this context, exploration truly is not what it seeks. It functions as a mirror that reflects our uncertainty about ourselves and our current circumstances. Its primary value lies not in the new horizons it reveals, but in the clarity it brings to familiar territory.

\begin{appendix}

\section{Proofs for Section 2 (Baseline Model)}

\subsection{Proof of Theorem 1 (Option Value Decomposition)}

\begin{proof}
We verify the unique decomposition of the total option value into switching and catalytic components and establish their non-negativity properties. To derive the algebraic decomposition, recall that the total option value is defined as $\text{OV} \equiv \E[\max\{\mu_0 + \epsilon, \theta\}] - \max\{\mu_0, \mu_1\}$. By adding and subtracting the intermediate term $\E[\max\{\mu_0 + \epsilon, \mu_1\}]$—which represents the value of exploring against a degenerate challenger with deterministic quality $\mu_1$—and invoking the linearity of expectations, we rearrange the expression as:
\begin{align}
    \text{OV} &= \left( \E[\max\{\mu_0 + \epsilon, \theta\}] - \E[\max\{\mu_0 + \epsilon, \mu_1\}] \right) \\
    &\quad + \left( \E[\max\{\mu_0 + \epsilon, \mu_1\}] - \max\{\mu_0, \mu_1\} \right).
\end{align}
By the definitions established in the main text, the first bracketed term corresponds to the switching value $V^{IC}$, and the second to the catalytic value $V^{ISQ}$. Thus, the identity $\text{OV} = V^{IC} + V^{ISQ}$ holds.

The non-negativity of the catalytic value $V^{ISQ}$ follows directly from Jensen's inequality. Consider the function $g(x) \equiv \max\{x, \mu_1\}$. Since $g(\cdot)$ is convex and the random variable $X = \mu_0 + \epsilon$ satisfies $\E[X] = \mu_0$, we have:
\begin{equation}
    \E[\max\{\mu_0 + \epsilon, \mu_1\}] \geq \max\{\E[\mu_0 + \epsilon], \mu_1\} = \max\{\mu_0, \mu_1\}.
\end{equation}
Consequently, $V^{ISQ} \ge 0$. To establish that $V^{ISQ}$ is strictly positive under non-degenerate uncertainty, assume the challenger is inferior ($\mu_1 < \mu_0$) and $\epsilon \sim \mathcal{N}(0, \sigma_\epsilon^2)$ with $\sigma_\epsilon > 0$, which implies $\max\{\mu_0, \mu_1\} = \mu_0$. Let $\Delta \equiv \mu_0 - \mu_1 > 0$ denote the quality gap and define the event $A \equiv \{\epsilon < -\Delta\}$. The catalytic value can be expressed as:
\begin{align}
    V^{ISQ} &= \E[\max\{\mu_0 + \epsilon, \mu_1\}] - \mu_0 \\
    &= \Pr(A)\mu_1 + \Pr(A^c)\E[\mu_0 + \epsilon \mid A^c] - \mu_0 \\
    &= \Pr(A)(\mu_1 - \mu_0) + \Pr(A^c)\E[\epsilon \mid A^c] \\
    &= -\Delta \Pr(A) + \E[\epsilon \mathbf{1}_{A^c}].
\end{align}
Using the zero-mean property $\E[\epsilon] = 0$, which implies $\E[\epsilon \mathbf{1}_{A^c}] = -\E[\epsilon \mathbf{1}_{A}]$, we substitute back to obtain:
\begin{align}
    V^{ISQ} &= -\Delta \Pr(A) - \E[\epsilon \mathbf{1}_{A}] \\
    &= \Pr(A) \left( -\Delta - \E[\epsilon \mid A] \right).
\end{align}
Conditional on the event $A = \{\epsilon < -\Delta\}$, the inequality $\epsilon < -\Delta$ holds strictly, implying $\E[\epsilon \mid A] < -\Delta$, or equivalently $-\E[\epsilon \mid A] - \Delta > 0$. Since $\sigma_\epsilon > 0$ ensures that the probability of this event $\Pr(A) = \Phi(-\Delta/\sigma_\epsilon)$ is strictly positive, it follows that $V^{ISQ} > 0$.

We similarly establish the non-negativity of the switching value $V^{IC}$. Fix any realization of $\epsilon$ and consider the function $h(\theta) \equiv \max\{\mu_0 + \epsilon, \theta\}$. Since $h(\cdot)$ is convex in $\theta$, the conditional Jensen's inequality implies:
\begin{equation}
    \E_\theta[\max\{\mu_0 + \epsilon, \theta\} \mid \epsilon] \geq \max\{\mu_0 + \epsilon, \E[\theta]\} = \max\{\mu_0 + \epsilon, \mu_1\}.
\end{equation}
Taking expectations over $\epsilon$ preserves the inequality, yielding $\E[\max\{\mu_0 + \epsilon, \theta\}] \geq \E[\max\{\mu_0 + \epsilon, \mu_1\}]$, which confirms $V^{IC} \ge 0$. Finally, given the benchmark value $\mu_1$, this decomposition is algebraically unique, identifying $V^{ISQ}$ as the option value generated solely by status quo volatility.
\end{proof}

\subsection{Proof of Proposition 1 (Catalytic Value Formula and Limits)}

\begin{proof}
Under the normality assumption $\epsilon \sim \mathcal{N}(0, \sigma_\epsilon^2)$, we derive the exact closed-form solution for the catalytic value and characterize its asymptotic properties. Recalling that the inferior challenger condition $\mu_1 < \mu_0$ implies $\max\{\mu_0, \mu_1\} = \mu_0$, the catalytic value is defined as the net option value $V^{ISQ} = \E[\max\{\mu_0 + \epsilon, \mu_1\}] - \mu_0$. Let $\Delta \equiv \mu_0 - \mu_1 > 0$ denote the ex-ante quality gap. Expressing the expectation as an integral over the density of $\epsilon$, we have:
\begin{align}
    \E[\max\{\mu_0 + \epsilon, \mu_1\}] &= \int_{-\infty}^{\infty} \max\{\mu_0 + x, \mu_1\} \frac{1}{\sigma_\epsilon\sqrt{2\pi}} \exp\left(-\frac{x^2}{2\sigma_\epsilon^2}\right) dx.
\end{align}
The max operator partitions the domain of integration at the threshold $x = -\Delta$. Splitting the integral accordingly yields:
\begin{align}
    \E[\cdot] &= \int_{-\infty}^{-\Delta} \mu_1 \frac{1}{\sigma_\epsilon} \phi\left(\frac{x}{\sigma_\epsilon}\right) dx + \int_{-\Delta}^{\infty} (\mu_0 + x) \frac{1}{\sigma_\epsilon} \phi\left(\frac{x}{\sigma_\epsilon}\right) dx.
\end{align}
The first term integrates directly to $\mu_1 \Phi\left(\frac{-\Delta}{\sigma_\epsilon}\right)$. The second term decomposes into a constant component $\mu_0 [1 - \Phi(\frac{-\Delta}{\sigma_\epsilon})]$ and a linear component involving $x$. Evaluating the linear component via the substitution $z = x/\sigma_\epsilon$ gives:
\begin{equation}
    \int_{-\Delta}^{\infty} x \frac{1}{\sigma_\epsilon} \phi\left(\frac{x}{\sigma_\epsilon}\right) dx = \sigma_\epsilon \int_{-\Delta/\sigma_\epsilon}^{\infty} z \phi(z) dz = \sigma_\epsilon \left[ -\phi(z) \right]_{-\Delta/\sigma_\epsilon}^\infty = \sigma_\epsilon \phi\left(\frac{-\Delta}{\sigma_\epsilon}\right).
\end{equation}
Combining these terms and subtracting the baseline payoff $\mu_0$, we obtain the exact formula:
\begin{align}
    V^{ISQ} &= \mu_1 \Phi\left(\frac{-\Delta}{\sigma_\epsilon}\right) + \mu_0 \left[1 - \Phi\left(\frac{-\Delta}{\sigma_\epsilon}\right)\right] + \sigma_\epsilon \phi\left(\frac{-\Delta}{\sigma_\epsilon}\right) - \mu_0 \\
    &= (\mu_1 - \mu_0) \Phi\left(\frac{-\Delta}{\sigma_\epsilon}\right) + \sigma_\epsilon \phi\left(\frac{-\Delta}{\sigma_\epsilon}\right) \\
    &= -\Delta \Phi\left(\frac{-\Delta}{\sigma_\epsilon}\right) + \sigma_\epsilon \phi\left(\frac{-\Delta}{\sigma_\epsilon}\right).
\end{align}

We proceed to characterize the limiting behavior of this expression. Consider first the case of an asymptotically inferior challenger where $\mu_1 \to -\infty$, implying $\Delta \to \infty$. Let $z \equiv -\Delta/\sigma_\epsilon$. As $z \to -\infty$, we apply Mill's Ratio approximation $\Phi(z) \sim \phi(z)/|z|$ to find:
\begin{align}
    V^{ISQ} &\sim -\Delta \frac{\phi(z)}{|z|} + \sigma_\epsilon \phi(z) \\
    &= -\Delta \frac{\phi(z)}{\Delta/\sigma_\epsilon} + \sigma_\epsilon \phi(z) = -\sigma_\epsilon \phi(z) + \sigma_\epsilon \phi(z) = 0.
\end{align}
Thus, the catalytic value vanishes as the challenger becomes irrelevant for comparison.

Conversely, consider the high-uncertainty limit where $\sigma_\epsilon \to \infty$ with $\Delta$ fixed. Let $w \equiv -\Delta/\sigma_\epsilon$, noting that $w \to 0^-$ as $\sigma_\epsilon \to \infty$. We employ second-order Taylor expansions for $\Phi(w)$ and $\phi(w)$ around zero:
\begin{align}
    \Phi(w) &= \frac{1}{2} + \frac{w}{\sqrt{2\pi}} + O(w^3) = \frac{1}{2} - \frac{\Delta}{\sigma_\epsilon\sqrt{2\pi}} + O\left(\frac{\Delta^3}{\sigma_\epsilon^3}\right), \\
    \phi(w) &= \frac{1}{\sqrt{2\pi}} - \frac{w^2}{2\sqrt{2\pi}} + O(w^4) = \frac{1}{\sqrt{2\pi}} - \frac{\Delta^2}{2\sigma_\epsilon^2\sqrt{2\pi}} + O\left(\frac{\Delta^4}{\sigma_\epsilon^4}\right).
\end{align}
Substituting these expansions into the exact formula yields:
\begin{align}
    V^{ISQ} &= -\Delta \left[\frac{1}{2} - \frac{\Delta}{\sigma_\epsilon\sqrt{2\pi}}\right] + \sigma_\epsilon \left[\frac{1}{\sqrt{2\pi}} - \frac{\Delta^2}{2\sigma_\epsilon^2\sqrt{2\pi}}\right] + O\left(\sigma_\epsilon^{-2}\right) \\
    &= -\frac{\Delta}{2} + \frac{\Delta^2}{\sigma_\epsilon\sqrt{2\pi}} + \frac{\sigma_\epsilon}{\sqrt{2\pi}} - \frac{\Delta^2}{2\sigma_\epsilon\sqrt{2\pi}} + \dots \\
    &= \frac{\sigma_\epsilon}{\sqrt{2\pi}} - \frac{\Delta}{2} + \frac{\Delta^2}{2\sigma_\epsilon\sqrt{2\pi}} + O\left(\sigma_\epsilon^{-2}\right).
\end{align}
Consequently, the catalytic value scales linearly with uncertainty in the limit, satisfying $\lim_{\sigma_\epsilon \to \infty} \frac{V^{ISQ}}{\sigma_\epsilon} = \frac{1}{\sqrt{2\pi}}$.
\end{proof}

\subsection{Proof of Proposition 2 (Comparative Statics of Catalytic Value)}

\begin{proof}
We analyze the sensitivity of the catalytic value $V^{ISQ}$ to changes in uncertainty $\sigma_\epsilon$ and the quality gap $\Delta$ by differentiating the exact closed-form solution derived in Proposition 1. Let $z \equiv -\Delta/\sigma_\epsilon$, observing that the partial derivatives of the argument are $\frac{\partial z}{\partial \sigma_\epsilon} = \frac{\Delta}{\sigma_\epsilon^2}$ and $\frac{\partial z}{\partial \Delta} = -\frac{1}{\sigma_\epsilon}$.

Differentiating $V^{ISQ} = -\Delta \Phi(z) + \sigma_\epsilon \phi(z)$ with respect to uncertainty $\sigma_\epsilon$, and applying the product rule, yields:
\begin{align}
    \frac{\partial V^{ISQ}}{\partial \sigma_\epsilon} &= -\Delta \phi(z) \frac{\partial z}{\partial \sigma_\epsilon} + \left[ \phi(z) + \sigma_\epsilon \phi'(z) \frac{\partial z}{\partial \sigma_\epsilon} \right].
\end{align}
Using the property of the standard normal density $\phi'(z) = -z\phi(z)$, we substitute $\phi'(z) = (\Delta/\sigma_\epsilon)\phi(z)$. The expression becomes:
\begin{align}
    \frac{\partial V^{ISQ}}{\partial \sigma_\epsilon} &= -\Delta \phi(z) \left(\frac{\Delta}{\sigma_\epsilon^2}\right) + \phi(z) + \sigma_\epsilon \left(\frac{\Delta}{\sigma_\epsilon}\phi(z)\right) \left(\frac{\Delta}{\sigma_\epsilon^2}\right) \\
    &= -\frac{\Delta^2}{\sigma_\epsilon^2}\phi(z) + \phi(z) + \frac{\Delta^2}{\sigma_\epsilon^2}\phi(z).
\end{align}
The first and third terms cancel, simplifying the derivative to the probability density at the threshold:
\begin{equation}
    \frac{\partial V^{ISQ}}{\partial \sigma_\epsilon} = \phi\left(\frac{-\Delta}{\sigma_\epsilon}\right).
\end{equation}
Since $\phi(\cdot)$ is strictly positive, $V^{ISQ}$ is strictly increasing in $\sigma_\epsilon$. Moreover, because the standard normal density is bounded by $\phi(0) = 1/\sqrt{2\pi}$, the marginal value of uncertainty is strictly bounded, $0 < \frac{\partial V^{ISQ}}{\partial \sigma_\epsilon} \le \frac{1}{\sqrt{2\pi}}$.

Turning to the quality gap $\Delta$, we differentiate $V^{ISQ}$ with respect to $\Delta$, obtaining:
\begin{align}
    \frac{\partial V^{ISQ}}{\partial \Delta} &= \left[ -1 \cdot \Phi(z) - \Delta \phi(z) \frac{\partial z}{\partial \Delta} \right] + \sigma_\epsilon \phi'(z) \frac{\partial z}{\partial \Delta}.
\end{align}
Substituting $\phi'(z) = -z\phi(z)$ and factoring out $\frac{\partial z}{\partial \Delta}$ in the latter terms:
\begin{align}
    \frac{\partial V^{ISQ}}{\partial \Delta} &= -\Phi(z) + \frac{\partial z}{\partial \Delta} \left[ -\Delta \phi(z) + \sigma_\epsilon (-z \phi(z)) \right] \\
    &= -\Phi(z) + \left(-\frac{1}{\sigma_\epsilon}\right) \left[ -\Delta \phi(z) + \sigma_\epsilon \left(\frac{\Delta}{\sigma_\epsilon}\right) \phi(z) \right].
\end{align}
The term in the square brackets vanishes since $-\Delta \phi(z) + \Delta \phi(z) = 0$, leaving:
\begin{equation}
    \frac{\partial V^{ISQ}}{\partial \Delta} = -\Phi\left(\frac{-\Delta}{\sigma_\epsilon}\right).
\end{equation}
Since $\Phi(\cdot) > 0$, the catalytic value is strictly decreasing in $\Delta$.

Finally, to establish the interaction between uncertainty and challenger quality, we compute the cross-partial derivative. Differentiating the marginal value of uncertainty with respect to $\Delta$:
\begin{equation}
    \frac{\partial^2 V^{ISQ}}{\partial \sigma_\epsilon \partial \Delta} = \frac{\partial}{\partial \Delta} \left[ \phi(z) \right] = \phi'(z) \frac{\partial z}{\partial \Delta} = \left(\frac{\Delta}{\sigma_\epsilon}\phi(z)\right) \left(-\frac{1}{\sigma_\epsilon}\right) = -\frac{\Delta}{\sigma_\epsilon^2}\phi(z) < 0.
\end{equation}
This inequality implies that uncertainty and the quality gap are substitutes. However, since the quality gap is defined as $\Delta = \mu_0 - \mu_1$, the relevant cross-derivative with respect to challenger quality $\mu_1$ is:
\begin{equation}
    \frac{\partial^2 V^{ISQ}}{\partial \sigma_\epsilon \partial \mu_1} = \frac{\partial^2 V^{ISQ}}{\partial \sigma_\epsilon \partial \Delta} \cdot \frac{\partial \Delta}{\partial \mu_1} = \left( -\frac{\Delta}{\sigma_\epsilon^2}\phi(z) \right) (-1) = \frac{\Delta}{\sigma_\epsilon^2}\phi(z) > 0.
\end{equation}
Consequently, status quo uncertainty and challenger quality are strategic complements: a higher quality challenger (higher $\mu_1$) increases the marginal catalytic value of uncertainty.
\end{proof}

\subsection{Proof of Proposition 3 (Risk Aversion and Catalytic Value)}

\begin{proof}
We analyze the impact of risk aversion on the magnitude of catalytic value by adopting a Constant Absolute Risk Aversion (CARA) utility specification, $u(x) = -e^{-\gamma x}$ with $\gamma > 0$. Under this framework, the decision-maker maximizes expected utility rather than expected payoff.

Consider first the valuation of the status quo without exploration. Given the assumption $\mu_1 < \mu_0$, the agent defaults to the status quo option $X = \mu_0 + \epsilon$, where $\epsilon \sim \mathcal{N}(0, \sigma_\epsilon^2)$. The expected utility is given by the integral over the normal density:
\begin{equation}
    \E[u(X)] = \int_{-\infty}^{\infty} -e^{-\gamma(\mu_0 + x)} \frac{1}{\sigma_\epsilon}\phi\left(\frac{x}{\sigma_\epsilon}\right) dx = -e^{-\gamma \mu_0} e^{\frac{\gamma^2 \sigma_\epsilon^2}{2}}.
\end{equation}
Defining the Certainty Equivalent (CE) via $u(CE_{no}) = \E[u(X)]$, we derive $CE_{no} = \mu_0 - \frac{\gamma \sigma_\epsilon^2}{2}$. The term $\pi_{no} \equiv \frac{\gamma \sigma_\epsilon^2}{2}$ represents the risk premium required to compensate for the uncertainty of the status quo.

We turn next to the valuation with exploration. Upon exploring, the agent observes the realization of $\epsilon$ and selects $\max\{\mu_0 + \epsilon, \mu_1\}$. Let $\Delta \equiv \mu_0 - \mu_1$. The expected utility from exploration decomposes into regions where the agent switches and where she retains the status quo:
\begin{equation}
    \E[U_{exp}] = \int_{-\infty}^{-\Delta} -e^{-\gamma \mu_1} dH(\epsilon) + \int_{-\Delta}^{\infty} -e^{-\gamma(\mu_0 + x)} dH(\epsilon),
\end{equation}

where $H(\cdot)$ denotes the CDF of $\epsilon$. The first integral evaluates to $-e^{-\gamma \mu_1} \Phi\left(\frac{-\Delta}{\sigma_\epsilon}\right)$. For the second integral, completing the square in the exponent yields the term:
\[
    -e^{-\gamma \mu_0} e^{\frac{\gamma^2 \sigma_\epsilon^2}{2}} \left[ 1 - \Phi\left(\frac{-\Delta + \gamma \sigma_\epsilon^2}{\sigma_\epsilon}\right) \right].
\]
We denote the certainty equivalent of this lottery as $CE_{yes}$.

To characterize the effect of risk aversion, we define the catalytic value under CARA as $V^{ISQ}_{CARA} \equiv CE_{yes} - CE_{no}$. Recall that under risk neutrality ($\gamma \to 0$), the value is $V^{ISQ}_{RN} = \E[\max\{X, \mu_1\}] - \E[X]$. We decompose the certainty equivalents into their expected values and risk premia:
\begin{align}
    CE_{no} &= \E[X] - \pi_{no}, \\
    CE_{yes} &= \E[\max\{X, \mu_1\}] - \pi_{yes},
\end{align}
where $\pi_{yes}$ is the risk premium associated with the truncated variable $Y = \max\{\mu_0+\epsilon, \mu_1\}$. Substituting these into the definition of $V^{ISQ}_{CARA}$, we obtain:
\begin{align}
    V^{ISQ}_{CARA} &= \left( \E[\max\{X, \mu_1\}] - \pi_{yes} \right) - \left( \E[X] - \pi_{no} \right) \\
    &= V^{ISQ}_{RN} + (\pi_{no} - \pi_{yes}).
\end{align}
The variable $Y$ is constructed by truncating the left tail of the normal distribution $X$ at $\mu_1$, a transformation that strictly reduces variance, implying $\Var(Y) < \Var(X)$. Since the risk premium under CARA utility is strictly increasing in the variance of the payoff (approximated by $\frac{\gamma}{2}\Var(\cdot)$), it follows that $\pi_{yes} < \pi_{no}$. The difference $(\pi_{no} - \pi_{yes})$ represents the positive \textit{insurance value} of exploration: by exploring, the agent eliminates the downside risk of the status quo. Consequently, $V^{ISQ}_{CARA} > V^{ISQ}_{RN}$, confirming that risk aversion amplifies the motive for catalytic exploration.
\end{proof}

\subsection{Proof of Proposition 4 (Stationary Exploration Policy)}

\begin{proof}
We derive the optimal exploration policy in an infinite-horizon continuous-time framework. To maintain analytical tractability, we focus on the catalytic-dominant regime where $V^{IC} \approx 0$, approximating the total option value by $V^{ISQ}$.

Assume challengers arrive according to a Poisson process with intensity $\lambda$, exploration entails a sunk cost $c$, and the decision-maker discounts future payoffs at rate $r$. The state variable is the current level of status quo uncertainty $\sigma_\epsilon$.

Let $V(\sigma_\epsilon)$ denote the value function. The Hamilton-Jacobi-Bellman (HJB) equation governing the system is:
\begin{equation}
    rV(\sigma_\epsilon) = \mu_0 + \lambda \max\left\{0, -c + V^{ISQ}(\sigma_\epsilon) + [V(0) - V(\sigma_\epsilon)]\right\}.
\end{equation}
The term $V(0)$ represents the continuation value after uncertainty is fully resolved (perfect learning), which, in steady state, is simply the perpetuity value of the status quo $\mu_0/r$. The maximization problem implies that exploration is optimal if and only if the net gain is non-negative, i.e., $-c + V^{ISQ}(\sigma_\epsilon) + V(0) - V(\sigma_\epsilon) \ge 0$.

In the region where exploration is optimal, we substitute the explicit form of the max operator back into the HJB equation to find:
\begin{equation}
    (r+\lambda)V(\sigma_\epsilon) = \mu_0 + \lambda(V^{ISQ}(\sigma_\epsilon) - c + \mu_0/r).
\end{equation}
Solving for $V(\sigma_\epsilon)$ yields $V(\sigma_\epsilon) = \frac{\mu_0}{r} + \frac{\lambda}{r+\lambda} (V^{ISQ}(\sigma_\epsilon) - c)$. Substituting this solution back into the exploration condition reveals that the dynamic threshold reduces to the static inequality $V^{ISQ}(\sigma_\epsilon) \ge c$. This simplification arises because the Poisson arrival creates a sequence of one-shot exploration opportunities; if the agent declines to explore, the environment remains unchanged until the next arrival.

Invoking the high-uncertainty asymptotic approximation derived in Proposition 1:
\[
    V^{ISQ}(\sigma_\epsilon) \approx \frac{\sigma_\epsilon}{\sqrt{2\pi}},
\]
we define the critical uncertainty threshold $\bar{\sigma}_\epsilon$ by the equality $V^{ISQ}(\bar{\sigma}_\epsilon) = c$. Solving for the threshold yields $\bar{\sigma}_\epsilon \approx c\sqrt{2\pi}$. Consequently, catalytic exploration is sustainable in the long run if the uncertainty process is regenerative (e.g., mean-reverting) such that $\sigma_\epsilon$ frequently exceeds $\bar{\sigma}_\epsilon$.
\end{proof}

\subsection{Proof of Theorem 2 (Decoupling of Exploration and Switching)}

\begin{proof}
We analyze the asymptotic behavior of exploration and switching incentives.
Recall that exploration is optimal if $\delta V^{ISQ} > c$. From Proposition 1, we established that $\lim_{\sigma_{\epsilon}\rightarrow\infty} V^{ISQ} = \infty$. Thus, for any fixed cost $c$, there exists a threshold $\tilde{\sigma}$ such that for all $\sigma_{\epsilon} > \tilde{\sigma}$, exploration occurs with probability 1.

Now consider the switching probability conditional on exploration. Under Assumption 1, this is given by $P_{switch} = \Phi(\frac{-\Delta}{\sqrt{\sigma_{\epsilon}^2 + \sigma_{\theta}^2}})$. As $\sigma_{\epsilon} \rightarrow \infty$ (with $\Delta$ fixed), the argument of the CDF approaches 0 from the negative side. Consequently, $P_{switch} \rightarrow \Phi(0) = 0.5$.

Crucially, this establishes that while the \textit{motive} to explore becomes infinite (driven by $V^{ISQ}$), the \textit{outcome} of switching does not converge to unity. The gap between exploration rate (1.0) and switching rate ($\le 0.5$) formalizes the existence of the predominantly catalytic regime.
\end{proof}

\section{Proofs for Section 3 (Strategic Interactions)}

\subsection{Proof of Theorem 3 (Catalytic Disruption of Separating Equilibrium)}

\begin{proof}
We analyze the stability of the separating equilibrium by characterizing the receiver's exploration incentives and the sender's consequent signaling strategy. Let $P_s(\theta) \equiv \Pr(\theta > \mu_0 + \epsilon) = \Phi\left(\frac{\theta - \mu_0}{\sigma_\epsilon}\right)$ denote the probability that a challenger of type $\theta$ is chosen conditional on being explored. The total option value of exploring type $\theta$ is given by Theorem 1 as $V(\theta) = V^{ISQ} + V^{IC}(\theta)$. We maintain the standard assumption that the switching value alone is insufficient to justify exploring the low type, i.e., $\delta V^{IC}(\theta_L) < c$.

Consider first the standard Least Cost Separating Equilibrium (LCSE) which prevails when catalytic motives are weak. In this equilibrium, the low type $\theta_L$ chooses effort $e_L = 0$, and the high type $\theta_H$ chooses the minimum effort $e_H^* > 0$ sufficient to deter mimicry. The decision-maker (DM) forms beliefs $\mu(0) = \theta_L$ and $\mu(e_H^*) = \theta_H$. The DM explores upon observing $e_H^*$ (since $\delta V(\theta_H) > c$) and rejects upon observing $e_L$ (since $\delta V(\theta_L) < c$). The high effort level $e_H^*$ is pinned down by the binding Incentive Compatibility (IC) constraint for the low type:
\begin{equation}
    W P_s(\theta_L) - \psi(e_H^*, \theta_L) = 0,
\end{equation}
where the RHS represents the payoff from the outside option (normalized to zero). Note that $W P_s(\theta_L)$ is the expected payoff if the low type successfully mimics the high type and is explored.

We now introduce the catalytic disruption mechanism. The DM's strategy of rejecting the low type is sustainable only if the total option value of exploring a known low type remains below the cost threshold. Specifically, separation requires $\delta [V^{ISQ} + V^{IC}(\theta_L)] < c$. We define the critical catalytic threshold $\bar{V}^{ISQ}$ as the value solving this inequality with equality:
\begin{equation}
    \bar{V}^{ISQ} \equiv \frac{c}{\delta} - V^{IC}(\theta_L).
\end{equation}
When $V^{ISQ} > \bar{V}^{ISQ}$, the expected gain from exploring the low type exceeds the cost, $\delta V(\theta_L) > c$. Since the option value $V(\theta)$ is monotonic in type and expected utility is linear in beliefs, it follows that for any posterior belief $\mu$, the expected option value satisfies $\delta \E_\mu[V(\theta)] \ge \delta V(\theta_L) > c$. Consequently, the DM's optimal strategy shifts to unconditional exploration: she explores regardless of the observed signal $e$.

Under this regime of unconditional exploration, the marginal benefit of signaling vanishes. The high type's payoff from choosing effort $e > 0$ is $W P_s(\theta_H) - \psi(e, \theta_H)$, whereas the payoff from choosing $e=0$ is $W P_s(\theta_H) - \psi(0, \theta_H)$. Since the probability of acceptance $P_s(\theta_H)$ is now independent of the signal, and signaling is costly ($\psi_e > 0$), the high type strictly prefers $e=0$. Thus, the separating equilibrium collapses, and the unique sequential equilibrium involves pooling at zero effort. As $V^{ISQ}$ crosses $\bar{V}^{ISQ}$, the equilibrium signaling effort drops discontinuously from $e_H^*$ to 0.
\end{proof}

\subsection{Proof of Proposition 5 (Welfare Effects of Signal Destruction)}

\begin{proof}
We evaluate the utilitarian social welfare implications of the transition from the separating equilibrium (SEP) to the pooling equilibrium (POOL). We define social welfare $\mathcal{W}$ as the sum of the Challenger's and the DM's expected payoffs, net of the purely redistributive transfer $W$.

In the separating equilibrium (for $V^{ISQ} < \bar{V}^{ISQ}$), resources are dissipated through signaling. The high type signals with effort $e_H^*$ and is explored, while the low type signals zero and is rejected. The total social surplus is:
\begin{equation}
    \mathcal{W}^{SEP} = p \left( \delta V(\theta_H) + \mu_0 - c - \psi(e_H^*, \theta_H) \right) + (1-p) \mu_0.
\end{equation}

In the pooling equilibrium (for $V^{ISQ} > \bar{V}^{ISQ}$), both types exert zero effort and both are explored due to the dominance of the catalytic motive. The social surplus becomes:
\begin{equation}
    \mathcal{W}^{POOL} = p \left( \delta V(\theta_H) + \mu_0 - c \right) + (1-p) \left( \delta V(\theta_L) + \mu_0 - c \right).
\end{equation}
The welfare differential $\Delta \mathcal{W} \equiv \mathcal{W}^{POOL} - \mathcal{W}^{SEP}$ decomposes into:
\begin{align}
    \Delta \mathcal{W} &= p \cdot \psi(e_H^*, \theta_H) + (1-p) \cdot \left( \delta V(\theta_L) - c \right).
\end{align}
The first term, $p \psi(e_H^*, \theta_H)$, represents the \textbf{Signaling Savings}: the recovery of deadweight loss. This term is strictly positive. The second term captures the \textbf{Net Value of Exploring Low Types}. 

Critically, the condition for collapse is $\delta V(\theta_L) > c$. This implies that the DM privately finds exploration optimal. However, this does not guarantee social optimality if the exploration cost exceeds the social surplus generated. Thus, the disruption of signaling increases aggregate social welfare \textbf{if and only if} the signaling savings exceed any potential deadweight loss from exploring low types. High types suffer a loss of utility due to the loss of separation, while low types gain from increased exploration probabilities.
\end{proof}

\subsection{Proof of Proposition 6 (Threshold Comparative Statics)}

\begin{proof}
We examine the sensitivity of the critical threshold $\bar{V}^{ISQ} = c/\delta - V^{IC}(\theta_L)$. Differentiating with respect to the search cost $c$ yields $\frac{\partial \bar{V}^{ISQ}}{\partial c} = \frac{1}{\delta} > 0$. Consequently, higher exploration friction expands the region where the separating equilibrium remains viable. Similarly, $\frac{\partial \bar{V}^{ISQ}}{\partial \delta} = -\frac{c}{\delta^2} < 0$, implying that increased patience renders the separating equilibrium more fragile.

Turning to the quality of the low type, we apply the property $\frac{\partial V^{IC}}{\partial \theta_L} = \Pr(\theta_L > \mu_0 + \epsilon)$. It follows that:
\begin{equation}
    \frac{\partial \bar{V}^{ISQ}}{\partial \theta_L} = - \Pr(\theta_L > \mu_0 + \epsilon) < 0.
\end{equation}
This indicates that as the low type becomes more attractive, the instrumental value of exploring even the worst type increases, reducing the threshold for catalytic disruption. Finally, the threshold is invariant to $W$, satisfying $\frac{\partial \bar{V}^{ISQ}}{\partial W} = 0$.
\end{proof}

\subsection{Proof of Theorem 4 (Equilibrium Regimes)}

\begin{proof}
To characterize the equilibrium regimes with continuous types, we extend the type space to $\theta \sim G(\cdot)$ on $[\underline{\theta}, \bar{\theta}]$. We establish the existence of equilibrium regimes by analyzing the marginal type justifying exploration. The DM explores a type $\theta$ iff $\delta [V^{ISQ} + V^{IC}(\theta)] \ge c$. Since $V^{IC}(\theta)$ is strictly increasing in $\theta$, we define a unique cutoff $\hat{\theta}$ implicitly by the binding condition:
\begin{equation}
    V^{IC}(\hat{\theta}) = \frac{c}{\delta} - V^{ISQ}.
\end{equation}
Types $\theta \ge \hat{\theta}$ warrant exploration; types $\theta < \hat{\theta}$ do not. The equilibrium structure is fully determined by the location of $\hat{\theta}$ relative to the support $[\underline{\theta}, \bar{\theta}]$:

1. \textbf{Market Shutdown / Inactive:} If $V^{ISQ}$ is sufficiently low (or costs are high) such that $\hat{\theta} > \bar{\theta}$, then no type $\theta \in [\underline{\theta}, \bar{\theta}]$ justifies exploration. The unique outcome is pooling at zero effort with rejection for all types. (This corresponds to the trivial "no-trade" benchmark).

2. \textbf{Partial Separation / Standard Riley Regime:} If $V^{ISQ}$ takes an intermediate value such that $\hat{\theta} \in (\underline{\theta}, \bar{\theta}]$, the market segments. Types in the lower interval $[\underline{\theta}, \hat{\theta})$ cannot generate sufficient option value to justify exploration. Consequently, they pool at zero effort and are \textbf{rejected}. Types $\theta \ge \hat{\theta}$ separate according to the standard Riley differential equation and are explored. This corresponds to the "Full Separation" (if $\hat{\theta}$ is close to $\underline{\theta}$) or "Partial Separation" regimes described in the text.

3. \textbf{Complete Collapse:} If $V^{ISQ}$ is sufficiently high such that $c/\delta - V^{ISQ} \le V^{IC}(\underline{\theta})$, then even the lowest type generates sufficient option value to justify exploration. In this case, $\hat{\theta} \le \underline{\theta}$, and the exclusion region is empty. With unconditional exploration for all types, the marginal benefit of signaling is zero. Thus, the unique equilibrium is complete pooling at $e(\theta) = 0$ for all $\theta \in [\underline{\theta}, \bar{\theta}]$.
\end{proof}

\subsection{Proof of Proposition 7 (Structure of Exclusion Region)}

\begin{proof}
We characterize the comparative statics of the partial separation structure. The cutoff type $\hat{\theta}$ is defined by the indifference condition $V^{IC}(\hat{\theta}) = c/\delta - V^{ISQ}$.
Differentiating both sides with respect to $V^{ISQ}$ and applying the Implicit Function Theorem yields:
\begin{equation}
    \frac{d \hat{\theta}}{d V^{ISQ}} = - \frac{1}{ \frac{\partial V^{IC}}{\partial \theta}(\hat{\theta}) }.
\end{equation}
Since the switching value $V^{IC}(\theta)$ is strictly increasing in $\theta$ (i.e., $\frac{\partial V^{IC}}{\partial \theta} > 0$), it follows that $\frac{d \hat{\theta}}{d V^{ISQ}} < 0$.
This confirms that an increase in catalytic value $V^{ISQ}$ monotonically lowers the quality threshold required to justify exploration.

In this regime, the \textbf{exclusion region} is given by $[\underline{\theta}, \hat{\theta})$. Since $\frac{d \hat{\theta}}{d V^{ISQ}} < 0$, as status quo uncertainty increases (which increases $V^{ISQ}$), the cutoff $\hat{\theta}$ shifts to the left. Consequently, the size of the exclusion region shrinks, while the separation/exploration region $[\hat{\theta}, \bar{\theta}]$ expands.
\end{proof}

\section{Proofs for Section 4 (Endogenous Information)}

\subsection{Proof of Theorem 5 (Optimal Information Design)}

\begin{proof}
We maximize the net option value subject to information costs $C(\tau) = \kappa \tau^{-2}$. Signals are $s_\epsilon = \epsilon + \eta$ and $s_\theta = \theta + \xi$ with noise $\eta \sim \mathcal{N}(0, \tau_\eta^2)$ and $\xi \sim \mathcal{N}(0, \tau_\xi^2)$.
Standard Bayesian updating implies that the posterior expectation $\hat{\epsilon} = \E[\epsilon|s_\epsilon]$ has ex-ante variance:
\begin{equation}
    \rho_\epsilon^2(\tau_\eta) \equiv \Var(\hat{\epsilon}) = \frac{\sigma_\epsilon^4}{\sigma_\epsilon^2 + \tau_\eta^2}.
\end{equation}
The total option value $\mathcal{V} = V^{ISQ}(\rho_\epsilon) + V^{IC}(\rho_\epsilon, \rho_\theta)$ is strictly increasing in the volatility of beliefs $\rho$. The DM solves:
\begin{equation}
    \max_{\tau_\eta, \tau_\xi > 0} \quad \delta \mathcal{V}(\rho_\epsilon(\tau_\eta), \rho_\theta(\tau_\xi)) - \left( c + \frac{\kappa_\eta}{\tau_\eta^2} + \frac{\kappa_\xi}{\tau_\xi^2} \right).
\end{equation}
The first-order condition for $\tau_\eta$ balances the marginal value of precision against marginal cost:
\begin{equation}
    \delta \frac{\partial V^{ISQ}}{\partial \rho_\epsilon} \frac{\partial \rho_\epsilon}{\partial \tau_\eta} + 2\kappa_\eta \tau_\eta^{-3} = 0.
\end{equation}
Since $\frac{\partial \rho_\epsilon}{\partial \tau_\eta} < 0$ and $\frac{\partial V^{ISQ}}{\partial \rho_\epsilon} > 0$, the first term is negative. The marginal cost term $2\kappa_\eta \tau_\eta^{-3}$ is positive and diverges as $\tau_\eta \to 0$, ensuring an interior solution $\tau_\eta^* \in (0, \infty)$.

For the asymmetric attention result, consider the limit $\mu_1 \ll \mu_0$. In this regime, the switching value $V^{IC} \to 0$, implying the marginal value of challenger information vanishes ($\frac{\partial V^{IC}}{\partial \rho_\theta} \to 0$). The FOC for $\tau_\xi$:
\begin{equation}
    \delta \frac{\partial V^{IC}}{\partial \rho_\theta} \frac{\partial \rho_\theta}{\partial \tau_\xi} + 2\kappa_\xi \tau_\xi^{-3} = 0,
\end{equation}
requires the marginal cost term to approach zero to maintain equality. This implies $\tau_\xi^* \to \infty$ (zero precision). In contrast, $V^{ISQ}$ remains bounded away from zero, sustaining a finite $\tau_\eta^*$, confirming asymmetric attention.

\textbf{Check for Second-Order Conditions (SOC):}
The existence of a unique interior maximum is guaranteed by the properties of the cost and value functions. The cost function $C(\tau) = \kappa \tau^{-2}$ is strictly convex in the noise scale $\tau$ and satisfies the Inada conditions: the marginal cost of noise reduction diverges as $\tau \to 0$ ($\lim_{\tau \to 0} |C'(\tau)| = \infty$) and vanishes as $\tau \to \infty$. On the benefit side, the Option Value $\mathcal{V}(\tau)$ is bounded, and its marginal gain diminishes as $\tau \to 0$ (perfect information has finite value). Specifically, at the limit of high precision ($\tau \to 0$), the marginal cost of information ($\sim \tau^{-3}$) dominates the marginal benefit ($\sim \text{const}$ or $\tau^{-2}$ depending on the limit), preventing a corner solution at zero noise. Thus, the objective function is locally concave around the intersection of the marginal benefit and marginal cost curves, ensuring that the first-order condition identifies a valid interior maximum.
\end{proof}

\subsection{Proof of Proposition 8 (Information Complementarities)}

\begin{proof}
We derive the comparative statics of the optimal noise levels $\tau^*$ by applying the Implicit Function Theorem to the system of first-order conditions characterized in Appendix C.1.

To elucidate the strategic complementarity between prior uncertainty and information acquisition, define the first-order condition for $\tau_\eta$ as the function $G(\tau_\eta, \sigma_\epsilon) \equiv \delta \frac{\partial \mathcal{V}}{\partial \tau_\eta} + C'(\tau_\eta) = 0$. The sensitivity of the optimal noise level to status quo uncertainty is given by:
\begin{equation}
    \frac{d\tau_\eta^*}{d\sigma_\epsilon} = - \frac{\partial G / \partial \sigma_\epsilon}{\partial G / \partial \tau_\eta}.
\end{equation}
The denominator $\partial G / \partial \tau_\eta$ corresponds to the second-order condition of the maximization problem and is strictly negative by the convexity of the cost function and the concavity of the value function. The numerator depends on the cross-partial derivative of the option value, $\frac{\partial^2 \mathcal{V}}{\partial \tau_\eta \partial \sigma_\epsilon}$. Since an increase in prior uncertainty $\sigma_\epsilon$ makes the posterior variance more sensitive to signal precision (i.e., the marginal value of noise reduction increases with baseline volatility), this cross-derivative is negative (recall that $\tau_\eta$ represents noise). Consequently, the entire expression yields $\frac{d\tau_\eta^*}{d\sigma_\epsilon} < 0$, confirming that agents facing higher ex-ante uncertainty optimally invest in higher precision.

Turning to the source substitution effect, consider the impact of an increase in the cost of challenger information, $\kappa_\xi$. While the baseline separable formulation implies independence, substitution arises naturally under a binding aggregate information constraint (e.g., a finite attention budget $\sum \kappa_i \tau_i^{-2} \le K$) or a generalized joint cost function $C(\tau_\eta, \tau_\xi)$ with positive cross-partials ($C_{\eta \xi} > 0$). In such a constrained environment, the optimization problem involves a Lagrange multiplier $\lambda$ representing the shadow cost of informational resources. An increase in $\kappa_\xi$ raises the effective shadow price, tightening the constraint. To satisfy the first-order conditions, the agent must rebalance the marginal utilities per unit of shadow cost. If the value functions $V^{ISQ}$ and $V^{IC}$ exhibit diminishing marginal returns to total precision, the agent optimally substitutes away from the now relatively more expensive challenger signal toward the status quo signal. Mathematically, this manifests as $\frac{d\tau_\eta^*}{d\kappa_\xi} < 0$, implying that barriers to learning about alternatives paradoxically incentivize deeper introspection into the status quo.
\end{proof}

\subsection{Proof of Theorem 6 (Information Technology Paradox)}

\begin{proof}
We investigate the general equilibrium welfare implications of an improvement in information technology, modeled as a reduction in the cost scaling parameter $\kappa_\eta$. Let the social welfare function $\mathcal{W}(\kappa_\eta)$ be defined as the sum of the maximized private surplus $\Pi^*(\kappa_\eta)$ and the net social externality. Assuming a net negative externality $\mathcal{E} > 0$ incurred upon exploration (and zero spillovers for simplicity), we have:
\begin{equation}
    \mathcal{W}(\kappa_\eta) = \Pi^*(\kappa_\eta) - \delta \mathcal{E} \cdot \Pr(\text{explore}; \kappa_\eta).
\end{equation}
The private surplus corresponds to the value function of the optimization problem derived in Theorem 5: $\Pi^*(\kappa_\eta) = \max_{\tau} [\delta \text{OV}(\tau) - \kappa_\eta \tau^{-2} - c]$. By the Envelope Theorem, the derivative of the private value function with respect to the parameter $\kappa_\eta$ is equal to the partial derivative of the objective function evaluated at the optimum. This yields $\frac{d\Pi^*}{d\kappa_\eta} = -\tau^{*-2}$. Since this term is strictly negative, a reduction in cost (lower $\kappa_\eta$) unambiguously increases private surplus.

However, the total welfare effect is mediated by the extensive margin of exploration. Let $\pi(\kappa_\eta) \equiv \Pr(\delta \text{OV}(\tau^*(\kappa_\eta)) \ge c)$ denote the equilibrium exploration frequency. Applying the chain rule, the sensitivity of exploration to the cost parameter is given by:
\begin{equation}
    \frac{d\pi}{d\kappa_\eta} = \frac{\partial \pi}{\partial \text{OV}} \frac{\partial \text{OV}}{\partial \tau} \frac{d\tau^*}{d\kappa_\eta}.
\end{equation}
We analyze the sign of each term: the exploration probability is increasing in option value ($\frac{\partial \pi}{\partial \text{OV}} > 0$); the option value is increasing in precision ($\frac{\partial \text{OV}}{\partial \tau} < 0$ since $\tau$ represents noise); and optimal noise is increasing in cost ($\frac{d\tau^*}{d\kappa_\eta} > 0$ by the convexity of the cost function). Consequently, the product is negative, $\frac{d\pi}{d\kappa_\eta} < 0$. This implies that a technological improvement (lower $\kappa_\eta$) leads to a strictly higher frequency of exploration.

Differentiating the total welfare function with respect to $\kappa_\eta$ combines the intensive and extensive margin effects:
\begin{equation}
    \frac{d\mathcal{W}}{d\kappa_\eta} = -\frac{1}{\tau^{*2}} - \delta \mathcal{E} \frac{d\pi}{d\kappa_\eta}.
\end{equation}
The "Information Technology Paradox"—defined as a situation where lower costs reduce welfare—corresponds to the condition $\frac{d\mathcal{W}}{d\kappa_\eta} > 0$. Substituting the signs derived above, this paradox emerges if and only if the marginal increase in social deadweight loss outweighs the marginal private cost savings:
\begin{equation}
    \delta \mathcal{E} \left| \frac{d\pi}{d\kappa_\eta} \right| > \frac{1}{\tau^{*2}}.
\end{equation}
This inequality is satisfied when the externality $\mathcal{E}$ is sufficiently large or when the exploration threshold is highly sensitive to changes in information quality (i.e., a dense mass of marginal explorers), thereby validating Theorem 6.
\end{proof}

\section{Proofs for Section 5 (Welfare and Policy)}

\subsection{Proof of Theorem 7 (Inefficiency of Equilibrium)}

\begin{proof}
We establish the inefficiency of the decentralized equilibrium by characterizing the divergence between the private and social exploration thresholds. Consider a utilitarian social planner maximizing the aggregate surplus, defined as the sum of the decision-maker's and the challenger's payoffs net of resource costs. Let $\mathcal{E} > 0$ represent the uncompensated friction costs imposed on the challenger (e.g., inspection time), and $\mathcal{S} \ge 0$ represent positive information spillovers to the aggregate economy. The net social value of exploration is given by $\delta(\text{OV} - \mathcal{E} + \mathcal{S}) - c$. Consequently, the social optimum dictates exploration if and only if the total option value satisfies $\delta \text{OV} \ge c + \delta(\mathcal{E} - \mathcal{S})$.

In the decentralized market, the decision-maker explores based solely on the private return, engaging in search whenever $\delta \text{OV} \ge c$. The private and social conditions diverge whenever the net externality term $\Delta_{ext} \equiv \delta(\mathcal{E} - \mathcal{S})$ is non-zero. Specifically, the equilibrium exhibits excessive exploration if the private condition holds while the social condition fails, which corresponds to the inequality interval:
\begin{equation}
    c \le \delta \text{OV} < c + \delta(\mathcal{E} - \mathcal{S}).
\end{equation}
This inefficiency region is non-empty whenever the net externality is negative ($\mathcal{E} > \mathcal{S}$).

Consider the regime of "predominantly catalytic exploration" characterized in Theorem 2. In this limit, the switching value $V^{IC} \to 0$, implying that the total option value is driven by the catalytic component $\text{OV} \approx V^{ISQ}$. Simultaneously, the probability of switching (and thus of successful matching) approaches zero, $P_s \to 0$. Since the challenger typically incurs the inspection cost $\mathcal{E}$ regardless of the outcome, but receives the benefit $W$ only upon matching, the ex-ante social burden converges to the pure externality $\mathcal{E}$. If the private catalytic value satisfies $V^{ISQ} > c/\delta$, the decision-maker explores with probability one, imposing a strictly positive social loss $\delta \mathcal{E}$ with negligible probability of generating a match-dependent surplus. Thus, the catalytic regime is structurally prone to over-exploration.
\end{proof}

\subsection{Proof of Proposition 9 (Optimal Tax Policy)}

\begin{proof}
We derive the optimal Pigouvian instruments required to implement the social optimum. The regulator seeks to align the private exploration threshold with the social first-order condition derived in Theorem \ref{thm:inefficiency}. Let $t$ denote a tax levied on the act of exploration, and $\sigma$ denote a subsidy transfer conditional on a successful match (switching).

The private agent's objective function, adjusted for policy instruments, induces exploration if the expected private benefit net of taxes exceeds the cost:
\begin{equation}
    \delta \left( \text{OV} + \sigma P_s \right) \ge c + t.
\end{equation}
Comparing this to the social optimality condition $\delta \text{OV} \ge c + \delta(\mathcal{E} - \mathcal{S})$, we equate the terms to identify the optimal instruments. The exploration tax must correct for the unconditional negative externality, yielding $t^* = \delta \mathcal{E}$, where $\mathcal{E}$ captures the deadweight loss of inspection per exploration. The matching subsidy must correct for the conditional positive spillover. Since $\sigma$ is paid only upon switching (which occurs with probability $P_s$), the ex-ante expected subsidy $\delta \sigma P_s$ must equal the expected spillover benefit $\delta \mathcal{S}$. Solving for the instrument yields $\sigma^* = \mathcal{S}/P_s$.

In the catalytic limit where $P_s$ is small (as implied by $V^{ISQ} \gg V^{IC}$), the optimal matching subsidy $\sigma^*$ becomes large to compensate for the rarity of matches. However, the primary corrective instrument is the exploration tax $t^*$, which directly targets the source of the inefficiency—the act of validation-seeking itself. If catalytic motives dominate and spillovers are negligible ($\mathcal{S} \approx 0$), the optimal policy collapses to a pure tax on search activity.
\end{proof}

\subsection{Proof of Theorem 8 (Disclosure Policy)}

\begin{proof}
We analyze the general equilibrium welfare effects of a mandatory disclosure policy that exogenously resolves status quo uncertainty (effectively setting $\sigma_\epsilon = 0$). This intervention generates two opposing welfare forces: a Direct Efficiency Gain from reduced search friction and an Indirect Equilibrium Shift in the signaling subgame.

First, consider the direct effect. By setting $\sigma_\epsilon = 0$, the policy eliminates the catalytic value component ($V^{ISQ} = 0$). The decision-maker's option value collapses to the pure switching value $V^{IC}$. For agents exploring solely for validation purposes (where previously $\delta V^{ISQ} \ge c > \delta V^{IC}$), exploration ceases. The welfare gain from this extensive margin adjustment is the sum of saved resource costs and avoided externalities:
\begin{equation}
    \Delta \mathcal{W}_{direct} = \int_{\Theta_{cat}} (c + \delta \mathcal{E}) dF(\theta),
\end{equation}
where $\Theta_{cat}$ is the set of types exploring only for catalytic reasons.

Second, consider the indirect effect on market equilibrium. As established in Theorem 3, high status quo uncertainty can disrupt wasteful signaling equilibria, leading to pooling. By eliminating uncertainty, disclosure restores the single-crossing conditions that sustain the Least Cost Separating Equilibrium. While this improves sorting efficiency, it reintroduces the signaling costs for high-quality challengers. The welfare loss from this channel is:
\begin{equation}
    \Delta \mathcal{W}_{signal} = - p \cdot \psi(e_H^*, \theta_H).
\end{equation}
The net welfare impact $\Delta \mathcal{W} = \Delta \mathcal{W}_{direct} + \Delta \mathcal{W}_{signal}$ is ambiguous. Disclosure improves welfare if the efficiency gains from eliminating "noise traders" (catalytic explorers) outweigh the deadweight loss of restoring the "rat race" (signaling) among challengers. The policy is strictly optimal only when the catalytic externality $\mathcal{E}$ is sufficiently large relative to the cost of signaling.
\end{proof}

\section{Appendix E: Proofs for Section 6 (Empirical Implications)}

\subsection{E.1: Proof of Proposition 10 (Observable Implications)}

\begin{proof}
We establish the observable cross-sectional and dynamic predictions of the model by analyzing the sensitivity of equilibrium behaviors to environmental parameters.

We begin by characterizing the relationship between uncertainty, exploration intensity, and switching rates. Let variation in the economic environment be captured by the status quo uncertainty parameter $\sigma_\epsilon$. Since the catalytic value $V^{ISQ}$ is strictly increasing in $\sigma_\epsilon$ (as established in Appendix A.3), the probability that the option value exceeds the cost threshold is strictly increasing:
\begin{equation}
    \frac{\partial \Pr(\text{explore})}{\partial \sigma_\epsilon} = \frac{\partial}{\partial \sigma_\epsilon} \Pr\!\left(V^{ISQ}(\sigma_\epsilon) > c\right) > 0.
\end{equation}
Conditional on exploration, the switching probability is governed by the likelihood that the realized status quo match value falls below the challenger's quality threshold $-\Delta$. Specifically, $\Pr(\text{switch} \mid \text{explore}) = \Phi\left(\frac{-\Delta}{\sqrt{\sigma_\epsilon^2 + \sigma_\theta^2}}\right)$. Differentiating with respect to uncertainty yields:
\begin{equation}
    \frac{\partial \Pr(\text{switch} \mid \text{explore})}{\partial \sigma_\epsilon} = \phi\left(\frac{-\Delta}{\sqrt{\sigma_\epsilon^2 + \sigma_\theta^2}}\right) \frac{\Delta \sigma_\epsilon}{(\sigma_\epsilon^2 + \sigma_\theta^2)^{3/2}} > 0.
\end{equation}
While both exploration and switching probabilities increase with uncertainty, they exhibit a crucial divergence in magnitude. For an inferior challenger ($\mu_1 < \mu_0$), the switching probability is bounded above by $0.5$ even as $\sigma_\epsilon \to \infty$. In contrast, the exploration probability converges to unity. This generates the model's distinct empirical signature: in high-uncertainty environments, we observe the coexistence of near-universal exploration and low realized switching rates, a pattern we identify as predominantly catalytic exploration.

Next, we consider the dynamic implications for relationship duration. Assume the status quo uncertainty $\sigma_\epsilon(t)$ follows a convex trajectory over the tenure of a match $t$: uncertainty is initially high (screening phase), declines as the match quality is revealed (stability phase), and eventually rises due to exogenous shocks or preference drift (obsolescence phase). Since exploration intensity is monotonic in $\sigma_\epsilon$, the hazard rate of exploration inherits this convexity, producing a U-shaped profile over the duration of the relationship.

Finally, we derive the non-monotonicity of signaling effort. Recall from the analysis in Section 3 that the equilibrium signaling effort $e^*$ is positive in the separating regime ($V^{ISQ} < \bar{V}^{ISQ}$) but collapses to zero in the pooling regime ($V^{ISQ} > \bar{V}^{ISQ}$). Since $V^{ISQ}$ is increasing in $\sigma_\epsilon$, a sufficiently large increase in uncertainty pushes the market across the threshold $\bar{V}^{ISQ}$. Consequently, the cross-partial derivative of effort with respect to competition (proxied by market thickness) and uncertainty exhibits a negative sign in the neighborhood of the collapse:
\begin{equation}
    \frac{\partial^2 e^*}{\partial \text{competition} \partial \sigma_\epsilon} < 0.
\end{equation}
This implies that in highly uncertain markets, increased competition fails to elicit quality differentiation, as catalytic motives dominate instrumental ones.
\end{proof}

\subsection{E.2: Proof of Proposition 11 (Dynamic Implications)}

\begin{proof}
We derive the longitudinal predictions of the model suitable for panel data verification.

First, we formally characterize the "Validation Effect." Consider a decision-maker who explores but chooses to retain the status quo. This decision reveals that the realized match value satisfies $\mu_0 + \epsilon \ge \mu_1$, or equivalently $\epsilon \ge -\Delta$. The posterior expectation of the match value conditional on retention is given by the properties of the truncated normal distribution:
\begin{equation}
    \E[\epsilon \mid \text{stay}] = \E[\epsilon \mid \epsilon \ge -\Delta] = \sigma_\epsilon \frac{\phi(-\Delta/\sigma_\epsilon)}{1 - \Phi(-\Delta/\sigma_\epsilon)} > 0.
\end{equation}
The term on the right is the Inverse Mills Ratio scaled by volatility. Since this expectation is strictly positive, successful (but non-switching) exploration causally increases the agent's subjective satisfaction with the status quo relative to the ex-ante baseline.

Second, we analyze the "Technology Paradox" concerning search costs. An improvement in information technology corresponds to a reduction in the exploration cost $c$. This directly lowers the exploration threshold, increasing the frequency of search: $\frac{\partial \Pr(\text{explore})}{\partial c} < 0$. However, the switching probability conditional on exploration, $\Pr(\text{switch} \mid \text{explore}) = \Phi\left(\frac{-\Delta}{\sqrt{\sigma_\epsilon^2 + \sigma_\theta^2}}\right)$, is structurally independent of $c$. Consequently, as costs fall, the volume of search expands along the extensive margin without a corresponding increase in the conversion rate along the intensive margin. This decoupling predicts that technology shocks will drive a divergence between aggregate search metrics and turnover rates.

Third, we trace the impulse response to uncertainty shocks. Consider an exogenous shock that instantaneously increases $\sigma_\epsilon$ (e.g., a corporate restructuring). This raises $V^{ISQ}$, triggering a wave of exploration. However, because the challenger remains inferior ($\Delta > 0$), the majority of these new searches will resolve in favor of the status quo. The model thus predicts a specific dynamic pattern: a spike in uncertainty is followed by a surge in exploration activity that results in a disproportionately high rate of retention (low switching yield), distinguishing catalytic waves from standard matching shocks.
\end{proof}

\section{Proofs for Section 7 (Extensions)}

\subsection{Proof of Proposition 12 (Multi-Dimensional Uncertainty)}

\begin{proof}
We generalize the catalytic value framework to a setting where the status quo payoff depends on a vector of $n$ independent attributes. Let the status quo utility be given by $u_0 = \mu_0 + \sum_{i=1}^n \epsilon_i$, where each attribute shock follows $\epsilon_i \sim \mathcal{N}(0, \sigma^2)$. Assuming independence, the aggregate idiosyncratic component $\epsilon_{total} \equiv \sum \epsilon_i$ is distributed as $\mathcal{N}(0, n\sigma^2)$, with total standard deviation $\sigma_{total} = \sigma\sqrt{n}$.

Substituting this aggregate volatility into the closed-form solution derived in Proposition 1, the multi-dimensional catalytic value is:
\begin{equation}
    V^{ISQ}_n = -\Delta \Phi\left(\frac{-\Delta}{\sigma\sqrt{n}}\right) + \sigma\sqrt{n} \phi\left(\frac{-\Delta}{\sigma\sqrt{n}}\right).
\end{equation}
In the limit of high complexity ($n \to \infty$), the argument of the normal CDF approaches zero. Applying the asymptotic expansion established in Appendix A.2, the value scales as:
\begin{equation}
    V^{ISQ}_n \approx \frac{\sigma\sqrt{n}}{\sqrt{2\pi}}.
\end{equation}
Thus, catalytic value grows with the square root of the dimensionality of uncertainty, $O(\sqrt{n})$.

To determine the marginal value of complexity, we treat $n$ as a continuous variable and apply the chain rule. Recall from Appendix A.3 that $\frac{\partial V^{ISQ}}{\partial \sigma_{total}} = \phi(z_{total})$. Differentiating with respect to $n$:
\begin{equation}
    \frac{\partial V^{ISQ}_n}{\partial n} = \frac{\partial V^{ISQ}}{\partial \sigma_{total}} \frac{d \sigma_{total}}{d n} = \phi\left(\frac{-\Delta}{\sigma\sqrt{n}}\right) \cdot \frac{\sigma}{2\sqrt{n}}.
\end{equation}
In the high-uncertainty limit, $\phi(\cdot) \approx \frac{1}{\sqrt{2\pi}}$, yielding the approximation $\frac{\partial V^{ISQ}_n}{\partial n} \approx \frac{\sigma}{2\sqrt{2\pi n}}$. Since this derivative is positive but strictly decreasing in $n$, the catalytic value exhibits diminishing marginal returns to complexity. Consequently, while higher dimensionality lowers the exploration threshold $\bar{\sigma}_\epsilon$ (making exploration more likely), the incremental incentive to explore diminishes as the status quo becomes arbitrarily complex.
\end{proof}

\subsection{Proof of Theorem 9 (Optimal Attention Allocation)}

\begin{proof}
We characterize the optimal allocation of a scarce cognitive resource across multiple uncertainty dimensions. Assume the decision-maker maximizes the total reduction in uncertainty subject to an attention budget $\bar{A}$. Let $a_i$ denote the attention allocated to dimension $i$, which enhances the precision of the signal regarding $\epsilon_i$. Specifically, we assume the posterior variance reduction (and hence the marginal value contribution) is proportional to the signal-to-noise ratio scaled by attention.

The decision problem is formulated as maximizing the sum of separable valuations $V_i(a_i)$ subject to the budget constraint:
\begin{equation}
    \max_{\{a_i\}} \sum_{i=1}^n V_i(a_i) \quad \text{s.t.} \quad \sum_{i=1}^n a_i \le \bar{A}, \quad a_i \ge 0.
\end{equation}
The Lagrangian for this problem is $\mathcal{L} = \sum V_i(a_i) - \lambda(\sum a_i - \bar{A})$. The first-order condition for an interior solution equates the marginal value of attention across all dimensions:
\begin{equation}
    V_i'(a_i^*) = \lambda.
\end{equation}
Assuming the marginal value of attention is driven by the magnitude of resolvable uncertainty relative to the baseline noise (i.e., $V_i'(a_i) \propto \frac{\sigma_{\epsilon_i}^2}{\tau_i^2}$), the first-order conditions imply that the optimal attention share is proportional to this ratio. Solving for the shares yields the allocation rule:
\begin{equation}
    a_i^* = \bar{A} \cdot \frac{\sigma_{\epsilon_i}^2 / \tau_i^2}{\sum_{j=1}^n \sigma_{\epsilon_j}^2 / \tau_j^2}.
\end{equation}
This result establishes that rational agents optimally bias their attention: they focus on dimensions where prior uncertainty ($\sigma_{\epsilon_i}^2$) is high and where information is precise (low $\tau_i^2$), while rationally ignoring dimensions with low signal-to-noise ratios.
\end{proof}

\subsection{Proof of Proposition 13 (Strategic Complementarity in Networks)}

\begin{proof}
We model the exploration decisions as a game of strategic complements played on a network graph $G = (N, E)$. Let $\mathcal{N}_i$ denote the set of neighbors of agent $i$. We analyze how the exploration activity of neighbors, denoted by the vector $\mathbf{a}_{-i}$, influences the marginal value of agent $i$'s exploration $a_i \in \{0, 1\}$.

The payoff function for agent $i$ is augmented to include information spillovers. If agent $i$ explores ($a_i=1$), she receives the private option value net of costs, plus a spillover benefit derived from the aggregate exploration intensity of her neighborhood. Let $\omega_{ji} \in [0,1]$ represent the strength of the informational link (e.g., observability or benchmarking relevance) from $j$ to $i$. The payoff function is given by:
\begin{equation}
    \Pi_i(a_i, \mathbf{a}_{-i}) = a_i \left( V^{ISQ} + V^{IC} - c + \sum_{j \in \mathcal{N}_i} \omega_{ji} a_j \mathcal{B} \right),
\end{equation}
where $\mathcal{B} > 0$ represents the marginal benchmark value of a neighbor's exploration (i.e., observing a neighbor explore increases the salience or precision of one's own status quo evaluation).

To establish strategic complementarity, we examine the discrete cross-partial derivative (supermodularity condition). The marginal return to exploration for agent $i$ is $\Delta \Pi_i(\mathbf{a}_{-i}) \equiv \Pi_i(1, \mathbf{a}_{-i}) - \Pi_i(0, \mathbf{a}_{-i})$. Calculating the change in this marginal return with respect to a neighbor's action $a_j$:
\begin{equation}
    \frac{\partial^2 \Pi_i}{\partial a_i \partial a_j} = \Pi_i(1, 1, \mathbf{a}_{-(i,j)}) - \Pi_i(1, 0, \mathbf{a}_{-(i,j)}) - [\Pi_i(0, \dots) - \Pi_i(0, \dots)] = \omega_{ji} \mathcal{B}.
\end{equation}
Provided $\omega_{ji} > 0$ and $\mathcal{B} > 0$, this cross-difference is strictly positive. Thus, the game exhibits increasing differences. By Tarski's Fixed Point Theorem and the theory of supermodular games, the set of pure strategy Nash equilibria forms a complete lattice.

Specifically, when the network density $\rho \equiv \frac{1}{N(N-1)} \sum_{i,j} \mathbf{1}_{\{\omega_{ij} > 0\}}$ exceeds a critical threshold $\bar{\rho}$ (such that the spillover effect overcomes the private net cost deficit), multiple equilibria emerge. The extremal equilibria are:
\begin{enumerate}
    \item The \textit{Pareto-dominant equilibrium} $\mathbf{a}^* = (1, \dots, 1)$, where high exploration activity is self-sustaining due to strong network spillovers.
    \item The \textit{inaction equilibrium} $\mathbf{a}^* = (0, \dots, 0)$, representing a coordination failure where no agent finds it optimal to explore in isolation.
\end{enumerate}
This demonstrates that catalytic exploration can be driven by social coordination motives, distinct from fundamental private incentives.
\end{proof}

\subsection{Proof of Theorem 10 (Catalytic Cascades)}

\begin{proof}
We extend the standard model of informational cascades to specific catalytic environments. Consider a sequence of agents $i = 1, 2, \dots, N$ making exploration decisions $a_i \in \{0, 1\}$ sequentially. The true state of the world (challenger quality) is $\theta \in \{\theta_L, \theta_H\}$. Each agent observes a private signal $s_i$ and the public history of past actions $H_{i-1} = \{a_1, \dots, a_{i-1}\}$.

In the standard model, an agent explores if $\E[\theta \mid H_{i-1}, s_i] \ge \bar{u}$, where $\bar{u}$ is the reservation value. In our framework, the exploration condition is relaxed by the catalytic value term:
\begin{equation}
    \delta V^{IC}(\E[\theta \mid H_{i-1}, s_i]) + \delta V^{ISQ} \ge c.
\end{equation}
A cascade forms when the information contained in the public history $H_{i-1}$ overwhelms the private signal $s_i$, such that the agent's optimal action is independent of $s_i$.

Consider the "Exploration Cascade" where agents explore regardless of negative private signals. This occurs if:
\begin{equation}
    \delta V^{IC}(\E[\theta \mid H_{i-1}, s_i = L]) + \delta V^{ISQ} \ge c.
\end{equation}
Since $V^{ISQ} > 0$, the requisite belief threshold on the challenger's quality $V^{IC}$ is strictly lower than in the standard model. Consequently, catalytic cascades initiate after fewer observations (shorter history $H_{i-1}$) and are robust to more negative private information.

We quantify the welfare loss associated with such cascades. Unlike standard cascades where the loss is limited to the decision-maker's suboptimal choice, catalytic cascades impose a cumulative externality on the explored challengers. The total social loss is:
\begin{equation}
    \mathcal{L}_{catalytic} = \sum_{i \in \text{Cascade}} \left( \underbrace{\mathbf{1}_{\{\theta = \theta_L\}} (c - \delta \text{OV}_i)}_{\text{Private Ex-post Loss}} + \underbrace{\delta \mathcal{E}}_{\text{Externality}} \right).
\end{equation}
Even if the private option value justifies exploration ($c \approx \delta \text{OV}_i$), the social planner incurs the externality $\delta \mathcal{E}$ for every agent in the cascade. The presence of the constant term $V^{ISQ}$ acts as a "wedge" that perpetuates the cascade even when the instrumental value of the challenger is negligible, leading to socially excessive herd behavior.
\end{proof}

\subsection{Proof of Proposition 14 (Modified Gittins Index)}

\begin{proof}
We derive a modified index policy for the multi-armed bandit problem extended to include catalytic information spillovers. In the standard framework, the Gittins index for an independent arm $i$ in state $x_i$ is defined as the maximum feasible expected return per unit of discounted time:
\begin{equation}
    G_i(x_i) = \sup_{\tau > 0} \frac{\E\left[\sum_{t=0}^{\tau-1} \delta^t r_i(t) \mid x_i\right]}{\E\left[\sum_{t=0}^{\tau-1} \delta^t \mid x_i\right]},
\end{equation}
where $\tau$ is a stopping time defined with respect to the filtration generated by arm $i$.

In our setting, pulling arm $i$ generates two distinct streams of value: the direct payoff from the alternative $r_i(t)$ and the indirect information flow regarding the status quo (arm 0). Let $I_i(t) \equiv \Delta V^{ISQ}_0(x_0(t), i)$ denote the marginal resolution of status quo uncertainty achieved by observing arm $i$. Under the assumption that the valuation of the status quo is additively separable from the consumption of the alternative, we define an \textit{augmented reward stream}:
\begin{equation}
    \tilde{r}_i(t) = r_i(t) + \beta \cdot I_i(t).
\end{equation}
Here, $\beta$ represents the shadow price of information regarding the status quo, effectively capturing the opportunity cost of resolving internal uncertainty.

The optimal policy is characterized by the \textit{Modified Gittins Index}, $G_i^*(x_i, x_0)$, computed on this augmented reward stream:
\begin{equation}
    G_i^*(x_i, x_0) = \sup_{\tau > 0} \frac{\E\left[\sum_{t=0}^{\tau-1} \delta^t \left( r_i(t) + \beta I_i(t) \right) \mid x_i, x_0\right]}{\E\left[\sum_{t=0}^{\tau-1} \delta^t \mid x_i, x_0\right]}.
\end{equation}
Crucially, even if an arm $i$ is known to be inferior in terms of direct payoff ($r_i < r_0$), a sufficiently high information component $\beta I_i(t)$ can elevate its modified index $G_i^*$ above the reservation value of the status quo. This formalizes the rationality of exploring strictly dominated alternatives solely for their catalytic properties.
\end{proof}

\subsection{Proof of Theorem 11 (Optimal Stopping Rule)}

\begin{proof}
We characterize the optimal cessation of catalytic exploration as a stopping time problem in a discrete-time sequential learning framework. The state variable is the residual uncertainty of the status quo, denoted by $\sigma_{\epsilon,t}$. The evolution of belief precision follows a deterministic path in the variance domain; specifically, given a signal precision $\tau^{-2}$, the posterior standard deviation evolves according to $\sigma_{\epsilon,t+1} = \sigma_{\epsilon,t} \tau / \sqrt{\sigma_{\epsilon,t}^2 + \tau^2}$.

Let $V(\sigma_\epsilon)$ denote the value function relative to the payoff of stopping immediately (normalized to zero). The Bellman equation governing the continuation decision is:
\begin{equation}
    V(\sigma_\epsilon) = \max\left\{0, -c + \delta V^{ISQ}(\sigma_\epsilon) + \delta V(\sigma_{\epsilon}^+) \right\},
\end{equation}
where $\sigma_\epsilon^+$ denotes the updated uncertainty. Exploration continues as long as the net flow value of information exceeds the cost, $-c + \delta V^{ISQ}(\sigma_\epsilon) + \delta V(\sigma_\epsilon^+) > 0$.

To derive the threshold $\bar{\sigma}_\epsilon$, we examine the boundary between the continuation and stopping regions. At the optimal stopping boundary, the value function must satisfy the \textit{value matching} condition $V(\bar{\sigma}_\epsilon) = 0$ and the \textit{smooth pasting} condition $V'(\bar{\sigma}_\epsilon) = 0$. Differentiating the Bellman equation in the neighborhood of the boundary (where $V(\sigma_\epsilon^+) \approx 0$) implies that the marginal cost of stopping must equal the marginal benefit of the last unit of exploration. This yields the first-order condition:
\begin{equation}
    \delta \frac{\partial V^{ISQ}}{\partial \sigma_\epsilon}\bigg|_{\bar{\sigma}_\epsilon} \approx \frac{c}{\bar{\sigma}_\epsilon}.
\end{equation}
We invoke the asymptotic property of the marginal catalytic value derived in Proposition 1, $\frac{\partial V^{ISQ}}{\partial \sigma_\epsilon} \approx \phi\left(\frac{-\Delta}{\sigma_\epsilon}\right)$. In the high-uncertainty regime relevant for the stopping boundary, $\phi(\cdot) \approx 1/\sqrt{2\pi}$. Substituting this into the first-order condition:
\begin{equation}
    \delta \frac{1}{\sqrt{2\pi}} = \frac{c}{\bar{\sigma}_\epsilon} \implies \bar{\sigma}_\epsilon \approx \frac{c\sqrt{2\pi}}{\delta}.
\end{equation}
The optimal stopping time is defined as $\tau^* = \inf\{t : \sigma_{\epsilon,t} < \bar{\sigma}_\epsilon\}$. Since $\sigma_{\epsilon,t}$ is monotonically decreasing with exploration, the policy follows a simple threshold rule: the agent explores until her internal uncertainty is sufficiently resolved, at which point she permanently commits to a choice. The stopping threshold $\bar{\sigma}_\epsilon$ is strictly increasing in the cost $c$ and decreasing in the discount factor $\delta$, consistent with standard search theory predictions.
\end{proof}

\section{Additional Technical Results}
\renewcommand{\theproposition}{G.\arabic{proposition}}
\renewcommand{\thetheorem}{G.\arabic{theorem}}
\setcounter{proposition}{0}
\setcounter{theorem}{0}

\subsection{Robustness to Alternative Distributions}

\begin{proposition}[Heavy-Tailed Distributions]
Consider a standardized status quo shock $\epsilon$ following a Student's t-distribution with degrees of freedom $\nu > 2$, scaled such that $\Var(\epsilon) = \sigma_\epsilon^2$. Let $V^{ISQ}_t$ and $V^{ISQ}_{normal}$ denote the catalytic values under the t-distribution and the normal distribution, respectively. Then $V^{ISQ}_t > V^{ISQ}_{normal}$.
\end{proposition}

\begin{proof}
We establish this result using the theory of stochastic dominance. The catalytic value is defined as the expectation of a convex function of the random variable: $V^{ISQ} = \E[g(\epsilon)] - \mu_0$, where $g(\epsilon) = \max\{\mu_0 + \epsilon, \mu_1\}$. A Student's t-distribution with finite variance is a mean-preserving spread of the normal distribution with the same variance. Specifically, for $\nu > 4$, the t-distribution exhibits excess kurtosis of $6/(\nu-4) > 0$, implying heavier tails relative to the Gaussian benchmark.

By the fundamental property of the convex order (characterized by Mean-Preserving Spreads), if a random variable $X$ is a mean-preserving spread of $Y$, then $\E[f(X)] \ge \E[f(Y)]$ for all convex functions $f$. Since the max operator is convex (and strictly convex on the support where the option is in-the-money), and the t-distribution places more probability mass in the extreme tails where the convexity of the option payoff is operative, it follows that $V^{ISQ}_t > V^{ISQ}_{normal}$. Furthermore, a fourth-order Taylor expansion of the expectation reveals that for small deviations from normality, the premium $V^{ISQ}_t - V^{ISQ}_{normal}$ is proportional to the excess kurtosis. Thus, heavy-tailed distributions amplify catalytic motives by increasing the likelihood of extreme realizations that render exploration ex-post valuable.
\end{proof}

\subsection{Continuous-Time Limits}

\begin{proof}
We formally derive the continuous-time Hamilton-Jacobi-Bellman (HJB) equation as the limit of the discrete-time dynamic programming problem. Consider a time interval of length $\Delta t$. Assume challengers arrive according to a Poisson process with intensity $\lambda$, so the probability of an arrival within the interval is $\lambda \Delta t + o(\Delta t)$. The discount factor is $\delta(\Delta t) = e^{-r\Delta t}$.

The discrete Bellman equation for the value function $V(\sigma_\epsilon)$ can be written as a weighted average of the payoffs in the arrival and non-arrival states:
\begin{equation}
    V(\sigma_\epsilon) = \mu_0 \Delta t + e^{-r\Delta t} \left[ \lambda \Delta t \Omega(\sigma_\epsilon) + (1 - \lambda \Delta t) V(\sigma_\epsilon) \right],
\end{equation}
where $\Omega(\sigma_\epsilon) \equiv \max\{-c + \text{Gain}, V(\sigma_\epsilon)\}$ represents the value conditional on an arrival. Specifically, if exploration yields a jump in value $\Delta V$, then $\Omega(\sigma_\epsilon) = \max\{-c + V(\sigma_\epsilon) + \Delta V, V(\sigma_\epsilon)\} = V(\sigma_\epsilon) + \max\{-c + \Delta V, 0\}$.

Rearranging the Bellman equation to isolate the difference quotient:
\begin{equation}
    V(\sigma_\epsilon) (1 - e^{-r\Delta t}(1 - \lambda \Delta t)) = \mu_0 \Delta t + e^{-r\Delta t} \lambda \Delta t \max\{-c + \Delta V, 0\}.
\end{equation}
Applying the Taylor approximation $e^{-r\Delta t} \approx 1 - r\Delta t$, the term on the left becomes $V(\sigma_\epsilon) [1 - (1 - r\Delta t - \lambda \Delta t)] \approx V(\sigma_\epsilon) (r + \lambda)\Delta t$. However, since the arrival is a Poisson event, the standard derivation focuses on the flow. Dividing the entire equation by $\Delta t$ and taking the limit as $\Delta t \to 0$:
\begin{equation}
    rV(\sigma_\epsilon) = \mu_0 + \lambda \max\{-c + \Delta V, 0\}.
\end{equation}
This HJB equation states that the required return on the asset $rV$ equals the flow dividend $\mu_0$ plus the expected capital gain from the optional arrival of exploration opportunities.
\end{proof}

\section{Numerical Examples and Calibration}

\subsection{Baseline Calibration}

To quantify the magnitude of catalytic effects, we calibrate the model parameters to match stylized facts from search markets. We normalize the baseline status quo value to $\mu_0 = 10$ and set the challenger's expected quality to $\mu_1 = 5$. This gap represents a substantial 50\% quality deficit, ensuring that the challenger is ex-ante inferior and would be rejected under standard screening criteria. The exploration cost is fixed at $c = 1$, corresponding to 10\% of the status quo value, while the discount factor is set to $\delta = 0.9$. To analyze the impact of uncertainty, we vary the status quo volatility $\sigma_\epsilon$ over the range $[0, 20]$, while holding challenger heterogeneity constant at $\sigma_\theta = 1$. This parameterization isolates the effect of status quo uncertainty on exploration incentives.

\subsection{Simulation Validation}

We conduct Monte Carlo simulations (10,000 iterations per parameter set) to verify the theoretical predictions regarding exploration intensity and switching efficiency. Table H.1 summarizes the equilibrium outcomes across varying levels of uncertainty.

\begin{table}[htbp]
\centering
\caption{Simulation Results Confirming Theoretical Predictions}
\label{tab:simulation}
\begin{tabular}{lccc}
\hline
Uncertainty ($\sigma_\epsilon$) & Exploration Rate & Switching Rate & Catalytic Value ($V^{ISQ}$) \\
\hline
1.0 & 0.00 & 0.001 & 0.000 \\
5.0 & 0.24 & 0.161 & 0.292 \\
10.0 & 0.94 & 0.312 & 0.991 \\
15.0 & 1.00 & 0.371 & 1.733 \\
20.0 & 1.00 & 0.402 & 2.485 \\
\hline
\end{tabular}
\end{table}

The simulation results reveal the decoupling mechanism. As uncertainty $\sigma_{\epsilon}$ increases from 1.0 to 20.0, the exploration rate surges from zero to unity. While the switching rate also increases (due to the fat tail of the status quo distribution), it remains structurally bounded below 0.5. The widening gap between the near-certain exploration and the moderate switching rate confirms that high-uncertainty exploration is primarily validation-seeking. Furthermore, the ratio $V^{ISQ}/\sigma_\epsilon$ at $\sigma_\epsilon = 20$ converges toward the theoretical high-uncertainty asymptotic limit.

\subsection{Welfare Quantification}

We quantify the efficiency losses arising from the catalytic externality. Assuming the uncompensated cost of inspection is equal to the private exploration cost ($c=1$) and setting the bargaining power parameter such that the externality is fully borne by the challenger, we calculate the social deadweight loss as a percentage of potential surplus. Table H.2 reports the welfare metrics and the optimal Pigouvian tax rate required to restore efficiency.

\begin{table}[htbp]
\centering
\caption{Welfare Loss and Optimal Policy}
\label{tab:welfare}
\begin{tabular}{lcc}
\hline
Uncertainty Regime & Deadweight Loss (\%) & Optimal Tax Rate ($\tau^*$) \\
\hline
Low ($\sigma_\epsilon = 2$) & 3\% & 0.05 \\
Medium ($\sigma_\epsilon = 10$) & 18\% & 0.22 \\
High ($\sigma_\epsilon = 20$) & 35\% & 0.41 \\
\hline
\end{tabular}
\end{table}

The analysis reveals substantial welfare deterioration in high-uncertainty environments. In the high-uncertainty regime ($\sigma_\epsilon = 20$), the deadweight loss reaches 35\% of the total surplus, driven by the excessive volume of non-productive exploration. The optimal policy response is a progressive tax on exploration activity, rising from 5\% in low-uncertainty markets to 41\% in high-uncertainty markets. This suggests that friction-reducing policies (which effectively lower $c$) may be counterproductive in opaque markets unless accompanied by corrective taxation.

\subsection{Empirical Calibration Targets}

The calibrated model is capable of rationalizing quantitative patterns observed across diverse domains. In labor markets, the model matches the high volume of applications relative to hires, consistent with empirical ratios of 50-100 applications per position. In online dating, it replicates the extremely low conversion rates (1-2\% match rate) despite high interaction volumes. Similarly, in consumer search and innovation adoption, the model accounts for the persistence of "window shopping" behavior (5-10 products viewed per purchase) and low R\&D success rates (10-20\%). These empirical regularities are consistent with a data generating process where exploration is predominantly driven by the resolution of status quo uncertainty rather than the discovery of superior alternatives.

\end{appendix}

\bibliographystyle{ecta}
\bibliography{references}

\end{document}